\documentclass[journal]{IEEEtran}

\usepackage{cite}

\usepackage{amsmath}
\usepackage{amssymb}
\usepackage{amsthm}
\usepackage{mathtools}
\usepackage{bm}
\usepackage{bbm}
\usepackage{xcolor}

\usepackage{array}
\usepackage{tikz}
\usepackage{pgfplots}
\pgfplotsset{compat=newest} 
\pgfplotsset{plot coordinates/math parser=false}
\usepackage{graphicx}

\usepackage{float}

\usetikzlibrary{decorations.markings}

\newtheorem{theorem}{Theorem}

\newtheorem{lemma}[theorem]{Lemma}
\newtheorem{corollary}[theorem]{Corollary}
\newtheorem{proposition}[theorem]{Proposition}
\newtheorem{remark}[theorem]{Remark}
\newtheorem{example}[theorem]{Example}

\newcommand{\feltstrom}{Feltstr{\"o}m}

\newcommand{\eBP}{\epsilon_{\text{BP}}^{\star}}

\newcommand{\SPBC}{{\rm\scriptscriptstyle  SPBC}}
\newcommand{\RBC}{{\rm\scriptscriptstyle  RBC}}
\newcommand{\BLEC}{{\rm\scriptscriptstyle  BLEC}}
\newcommand{\B}{{\rm\scriptscriptstyle  B}}
\newcommand{\pr}[1]{{\mathbb{P}}\left\{{#1}\right\}}
\newcommand{\eps}{\varepsilon}

\usetikzlibrary{external} 
\newcommand{\myfigure}[1]{\resizebox{\columnwidth}{!}{\includegraphics{figures_final/#1.pdf}}}
\newcommand{\myfigurenoresize}[1]{{\includegraphics{figures_final/#1.pdf}}}

\begin{document}

\title{Finite-Length Analysis of Spatially-Coupled Regular LDPC Ensembles on Burst-Erasure Channels}

\author{\IEEEauthorblockN{Vahid~Aref,~\IEEEmembership{Member,~IEEE,}
		Narayanan~Rengaswamy,~\IEEEmembership{Student~Member,~IEEE,}
		and~Laurent~Schmalen,~\IEEEmembership{Senior~Member,~IEEE}$^*$}%
\thanks{$^*$Authors are listed in alphabetical order. Parts of this paper have been presented at the 2016 International Zurich Seminar on Communications~\cite{RengaswamyIZS16} and the 2016 International Symposium on Turbo Codes and Iterative Decoding~\cite{ARS_ISTC16}.}
\thanks{V. Aref and L. Schmalen are with Nokia Bell Labs, Stuttgart 70435, Germany
 (e-mail: \texttt{\{first.last\}@nokia-bell-labs.com}).}%
\thanks{N. Rengaswamy is with the Department
of Electrical and Computer Engineering, Duke University, Durham,
NC, 27708 USA (e-mail: \texttt{narayanan.rengaswamy@duke.edu}).}
\thanks{This work was partially conducted while N. Rengaswamy was visiting Bell Labs as a research intern funded by a scholarship of the DAAD-RisePro programme. The work of V. Aref and L. Schmalen has been performed in the framework of the CELTIC
EUREKA project SENDATE-TANDEM (Project ID C2015/3-2), and it is
partly funded by the German BMBF (Project ID 16KIS0450K).}%
\thanks{Copyright~\copyright~2018 IEEE. Personal use of this material is permitted.  However, permission to use this material for any other purposes must be obtained from the IEEE by sending a request to \texttt{pubs-permissions@ieee.org}.}%
}

\maketitle

\begin{abstract}

Regular Spatially-Coupled LDPC (SC-LDPC) ensembles have gained significant interest since they were shown to universally achieve the capacity of binary memoryless channels under low-complexity belief-propagation decoding. 
In this work, we focus primarily on the performance of these ensembles over binary channels affected by bursts of erasures. We first develop an analysis of the finite length performance for a single burst per codeword and no errors otherwise. We first assume that the burst erases a complete spatial position, modeling for instance node failures in distributed storage. We provide new tight lower bounds for the block erasure probability ($P_\B$) at finite block length and bounds on the coupling parameter for being asymptotically able to recover the burst.
We further show that expurgating the ensemble can improve the block erasure probability by several orders of magnitude. Later we extend our methodology to more general channel models. In a first extension, we consider bursts that can start at a random location in the codeword and span across multiple spatial positions. 
Besides the finite length analysis, we determine by means of density evolution the maximum correctable burst length. In a second extension, we consider the case where in addition to a single burst, random bit erasures may occur. Finally, 
 we consider a block erasure channel model which erases each spatial position independently with some probability $p$, potentially introducing multiple bursts simultaneously. All results are verified using Monte-Carlo simulations.

\end{abstract}


\section{Introduction and Motivation}
\label{sec:intro}

Low-density parity-check (LDPC) codes, first introduced by Gallager in 1962~\cite{GallagerIT62}, are graph-based codes that have found widespread use due to their excellent performance under iterative belief-propagation (BP) decoding.
However, BP decoding is suboptimal and does not reach the performance of optimal maximum-a-posteriori (MAP) decoding~\cite{Richardson-MCT08}.
Yet, it is commonly employed in practice due to its significant computational advantage over MAP decoding.
The BP decoding threshold of LDPC codes can be improved towards values close to the capacity of the channel by the use of irregular LDPC codes~\cite{Richardson2001design}. 
However, capacity-approaching irregular LDPC codes usually require a large fraction of degree-2 variable nodes, which leads to undesirable finite-length properties like small minimum distance and a large amount of stopping or trapping sets that impair the BP decoding performance.

An attractive possibility to overcome this deficiency is the use of terminated spatially-coupled (SC) LDPC codes---originally introduced as convolutional LDPC ensembles by \feltstrom~and Zigangirov in~\cite{Felstrom-it99}---which show significantly better BP thresholds than LDPC codes without requiring large fractions of degree-2 variable nodes. It was numerically observed and conjectured~\cite{lentmaier2010thresholds} that the BP threshold of terminated SC-LDPC codes saturates to the MAP threshold of the underlying LDPC code ensemble. This phenomenonm, termed as \emph{threshold saturation}, was subsequently rigorously proven in~\cite{Kudekar-it11} and it was shown that this method allows us to asymptotically achieve capacity on the binary erasure channel (BEC) under low-complexity BP decoding with regular SC-LDPC ensembles. Threshold saturation was later shown to be universally true for any binary memoryless symmetric (BMS) channel~\cite{KudekarIT13}, and the result has triggered a lot of research interest for SC-LDPC codes and their practical applications.
While the asymptotic behavior of SC-LDPC codes is now well understood for a few years, recently the finite length performance of various constructions of SC-LDPC codes has been studied~\cite{Olmos-isit11,Stinner-isit14,Olmos-it15} and scaling laws to predict the finite length behavior have been proposed.

In this paper, we investigate the asymptotic as well as the finite-length behavior of regular SC-LDPC codes when a single burst or a sporadic random bursts of erasures occur on the channel. 
Burst erasures can model some common practical communication scenarios such as, e.g., deep fades in wireless communications, the failure of nodes in distributed storage systems and content delivery networks (CDNs), collisions in code slotted ALOHA~\cite{liva2012spatially}, or the loss of packets in packet-based communication systems, to name just a few.
One particularly interesting case is the erasure of a complete spatial position (SP), which can happen for instance when a node fails in distributed noisy storage, where every node is mapped to an SP. 
In such scenarios, many of the common practically relevant protograph-based constructions of SC-LDPC codes fail by construction. 
For example, if a complete SP is erased in the protograph construction of~\cite[Fig.~3]{Kudekar-it11}, all check nodes that have connections to this SP have at least two erased connections and hence, the erased variable nodes of that SP are not recoverable. 
This situation has been studied in~\cite{Iyengar-icc10}, where the authors additionally provide protograph constructions that avoid this situation and also maximize the correctable burst length given some structural constraints of the code. 
To correct bursts encompassing an SP, the authors in~\cite{Mori-corr15} apply interleaving (therein denoted band splitting) to a protograph-based SC-LDPC code. 
Interleaving however increases system latency and if windowed decoding is used, this approach results in an increased required window length and thus complexity. 

Recently, in~\cite{ulHassan-isit14},\cite{ulHassan-itw15}, it has been shown that some well-designed protograph-based LDPC codes can increase the diversity order of block fading channels and are thus also good candidates for block erasure channels. 
The tradeoff with these constructions is, however, that they require large syndrome former memories if the burst length becomes large. 
In addition, closely related structures based on protographs have been proposed in~\cite{Jardel-comnet15,Dedeoglu-ict15}, which spatially couple the previously proposed root-check LDPC codes~\cite{Boutros-tit10} to improve the finite length performance and thresholds. 
The analysis of the randomly coupled SC-LDPC ensemble of~\cite{Kudekar-it11}, for the case where complete SPs are randomly erased, has been discussed in~\cite{Jule-isit13}. In this work, the authors notice the robustness of general SC-LDPC ensembles against bursts and use the randomly erased SPs as a model for block-fading channels. 
Based on an asymptotic analysis, relatively loos easymptotic lower and upper bounds for the bit and block erasure probabilities are derived in~\cite{Jule-isit13}.

In~\cite{Jule-isit13} and some of the other previously mentioned works, the transmitted bits are either received without error, or erased by the burst. 
Additionally, the burst erasure affects complete SPs. 
In \cite{ARS_ISTC16}, we considered a more general model. 
First we allowed the burst erasure to take on any length and assumed that its starting position can occur anywhere in the codeword. 
Bursts of random starting positions have also been investigated in~\cite{AndriyanovaBlackSea15}, although for different reasons and applications (synchronization). 
Additionally, we assumed that the parts of the codeword not affected by the burst are affected by a memoryless noise process. 
We used density evolution to find the maximum correctable burst length when a random, regular SC-LDPC ensemble is used for transmission over the BEC or the binary additive white Gaussian noise channel (BiAWGN). 
We empirically observed that the correctable burst length is minimal when the starting position of the burst is exactly at the boundary of an SP. 
This means that the burst is less likely to be recovered when the first affected SP is completely erased. 
This fact additionally motivates the analysis of the scenario in which a burst erases exactly one SP.

Our main focus in this paper is the finite-length (non-asymptotic) analysis of the random regular SC-LDPC ensemble~\cite{Kudekar-it11} over channels with burst erasures. 
We consider the random regular SC-LDPC ensemble because it is universally capacity-achieving over a wide range of memoryless channels~\cite{KudekarIT13} and we investigate if this ensemble can be beneficial as well if burst erasures occur in the channel. 
We start with the case where a complete SP is erased and derive a necessary condition on the required coupling width 
to recover from the burst. 
Then we give a lower bound on the block erasure probability for finite length codes based on a novel stopping set analysis. 
Subsequently we study the effects of expurgation, i.e., the removal of short cycles from the graph. We show that expurgation significantly leads to better performance in the finite length regime when burst erasures occur.
We then generalize these results to situations where the burst does not necessarily occur at the boundary of an SP and can span more than one SP.
Based on density evolution, we first derive a bound that relates the coupling width $w$ and the length of bursts, and then focus again on the finite block length regime. 
Additionally, we introduce random erasures in the parts of the codeword not affected by the burst, as a more realistic model for noisy distributed storage or more traditional communication schemes. 

Finally, we consider a more general non-single-burst scenario where the channel erases each spatial position independently with some probability.
We demonstrate how our analyses on the single-burst channel models can be used to closely estimate the decoding failure probability on this channel as well.

The paper is organized as follows: 
In Section~\ref{sec:prelims}, we review the essential technical background, 
motivate our problem and define the channel models. 
In Section~\ref{sec:spbcAnalysis}, we introduce the novel finite-length analysis of the random ensemble for the case when a complete SP is erased. 
Then, in Section~\ref{sec:expurgate}, we detail the effects of expurgating the ensemble on this channel.
In section~\ref{sec:general_channels}, we extend the finite-length analysis to some other  channel models including: a single burst of erasures of arbitrary length and starting position inside the codeword, the transmission of the bits over a binary erasure channel (BEC) before a burst erasure occurs, and finally 
the block erasure channel which erases each spatial position independently with some probability.
Finally, in Section~\ref{sec:conclusions} we conclude the paper highlighting directions for future research.

\section{Preliminaries \& Motivation}
\label{sec:prelims}

\subsection{Notation}
We use $\mathbb{N}$ to denote the set of positive integers and $[n]$ to denote the set $[n] = \{1,2,\ldots,n\}$. 
In a graph, we denote the neighborhood of a vertex $v_i$ by $\mathcal{N}(v_i)$, i.e., $\mathcal{N}(v_i)$ is the set of all nodes such that there exists an edge between $v_i$ and the node. 
We say that, for $n\in\mathbb{N}$, a function $f(n)$ is $O(g(n))$ if there exists an $n_0\in\mathbb{N}$ and a positive constant $\gamma$ such that $|f(n)| \leq \gamma|g(n)|$ for all integers $n>n_0$. For two real values $a$ and $b$, we say $a\lessapprox b$ if $a\leq b$ and $\frac{b}{a}\approx 1$, i.e. $b$ is a good approximation of $a$. Similarly, we define the sign $\gtrapprox$. 

The symmetric binary erasure channel with parameter $\epsilon$ is denoted by BEC($\epsilon$) and its transition probabilities are defined, for $x\in \{0,1\}$ and $y\in \{0,1,?\}$, where $?$ denotes an erasure, by
\begin{equation*}
W(y|x)=
\begin{cases}
1-\epsilon & \text{if } y=x\\
\epsilon   & \text{if } y=?
\end{cases} .
\end{equation*}

\subsection{The Random Regular SC-LDPC Ensemble}
\label{sec:randomSCLDPC}

We now briefly review how to sample a code from the random regular SC-LDPC ensemble~\cite{Kudekar-it11}, denoted as $\mathcal{C}_{\mathcal{R}}(d_v,d_c,w,L,M)$. 
We first lay out a set of positions indexed from $z=1$ to $L$ on a \emph{spatial dimension}. 
At each spatial position (SP) $z$, there are $M$ variable nodes (VNs) and $M\frac{d_v}{d_c}$ check nodes (CNs), 
where $M\frac{d_v}{d_c} \in \mathbb{N}$ and, $d_v$ and $d_c$ denote the variable and check node degrees, respectively.
Let $w>1$ denote the coupling (smoothing) parameter. 
Then, we additionally consider $w-1$ sets of $M\frac{d_v}{d_c}$ CNs in SPs $L+1,\dots,L+w-1$. 
Every CN is equipped with $d_c$ \emph{sockets} and imposes an even parity constraint on its $d_c$ neighboring VNs. 
Each VN in SP $z$ is connected to $d_v$ CNs from SPs $z,\dots,z+w-1$ as follows: 
each of the $d_v$ edges of this VN is allowed to randomly and uniformly connect to any of the $wMd_v$ free sockets arising from the CNs in SPs $z,\dots,z+w-1$, such that \emph{parallel edges} are avoided in the resultant bipartite graph. 
When a VN has multiple edges connected between itself and a particular CN, those edges are called parallel edges.
We avoid parallel edges as it is well known that for finite $M$, the presence of parallel edges can have detrimental effects on the decoding performance~\cite{Richardson-MCT08}.
This graph represents the code so that we have $N=LM$ code bits, distributed over $L$ SPs.   
Note that the CNs at the boundaries, i.e., at SPs $1,\ldots,w-1$ and $L+1,\ldots,L+w-1$, can have degree less than $d_c$, due to termination of the code. 
CNs of degree zero are removed from the code.
Because of the additional CNs in SPs $z>L$, the code rate is $r = 1-\frac{d_v}{d_c}-O(\frac{w}{L})$.
Throughout this work, we assume the two mild conditions of $d_v \geq 3$ and $wM \geq 2(d_v+1)d_c$ (see Appendix~\ref{sec:appProof1}).

A subset $\mathcal{A}$ of  VNs in a code is a \emph{stopping set} if all the neighboring CNs of (the VNs in) $\mathcal{A}$ connect to $\mathcal{A}$ at least twice~\cite[Def. 3.137]{Richardson-MCT08}. 
In such a case, if all VNs in $\mathcal{A}$ have been erased by the channel, then the BP decoder will fail since all the neighboring CNs are connected to at least two erased VNs. 
Therefore, such a set will stop the decoding process and hence is called a \emph{stopping set}. 
The cardinality of the set $\mathcal{A}$ is also its size.
A \emph{minimal stopping set} is one which does not contain a smaller size non-empty stopping set within itself. 

\subsection*{The Poisson SC-LDPC Ensemble ($\mathcal{C}_\mathcal{P}$)}

For the sake of comparison, we also consider another ensemble. 
We still assume a regular variable node degree $d_v$, but there is no limit placed on the check degree. 
For each VN at spatial position $z$, we assume that its edges can connect to any check node at spatial positions $z,\ldots, z+w-1$ such that parallel edges are avoided (and without constraining the check node degree). 
Check nodes are selected uniformly at random for each edge. 
This construction yields the so-called Poisson ensemble $\mathcal{C}_{\mathcal{P}}(d_v,d_c,w,L,M)$, where $d_c$ specifies the \emph{average} check node degree of this ensemble. 
By the described construction procedure, we do not have a regular check node degree but instead an irregular distribution that follows a binomial distribution (which converges to a Poisson distribution for $M\to\infty$). 
By fixing the number of CNs per SP to $M\frac{d_v}{d_c}$, we indeed get an average check node degree of~$d_c$.

\subsection{Motivation}\label{sec:betamax}

Consider the random regular SC-LDPC ensemble $\mathcal{C_R}(d_v,d_c,w,L,M)$ affected by a single burst of $b$ bits. Let $S\in[M]$ denote the random starting VN of the burst at some SP $z_0$.
Note that $z_0$ is arbitrary\footnote{When there is only a burst of erasures, the performance will be equivalent to the performance of the circular tail-biting ensemble with the same parameters~\cite{tavares2007tail}.}. We define the normalized quantities $s \triangleq \frac{S}{M}$ and $\beta \triangleq \frac{b}{M}$, so that $0 \leq s \leq 1$ and $0 \leq \beta \leq L$.
Using density evolution in the limit of $M$, we showed numerically in \cite{ARS_ISTC16} 
that for any fixed $s$, there is a maximum normalized $\beta(s)$ recoverable with an arbitrarily small probability of decoding failure. 
Then, a burst of length $\beta M$ and random starting position is recoverable in the limit of $M$ if 
\begin{equation*}
\beta<\beta_{\rm max} = \min_{0\leq s< 1} \beta(s).
\end{equation*} 
Figure~\ref{fig:beta} illustrates $\beta(s)/\beta_{\rm max}$ as a function of the normalized starting position $s$ for the $\mathcal{C}_{\mathcal{R}}(d_v=3,d_c=6,w,L,M)$ SC-LDPC ensemble with $w\in\{3,4,5\}$. 
Interestingly, we observed that when $\beta_{\rm max}>1$, $s=0$ (the full erasure of the position $z_0$) is the worst case scenario as $\beta(0)$ is a minimum of $\beta(s)$. 
We will also see later in Example~\ref{ex:b=M} that the error floor is also larger when $s=0$. 
\begin{figure}[t!]
\newlength{\picwidth}
\newlength{\picheight}
\setlength{\picwidth}{0.4\textwidth}
\setlength{\picheight}{0.25\textwidth}
\centering
\myfigure{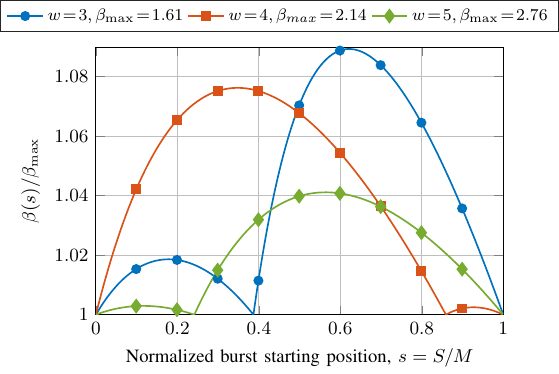}
\caption{The normalized maximum correctable burst length for a given starting position $S=sM$  when the $C_{\mathcal{R}}(3,6,w,L,M)$ ensemble is used.}\label{fig:beta}
\end{figure}

The observations of Fig.~\ref{fig:beta} motivates our first focus on a very simple channel model, where the starting position of a burst is $S=1$ and the burst length is exactly $b=M$. Later, we will extend our results to arbitrary values of $S$ and $b$ and then, to the case of occurring more than a single-burst. 
We refer the interested reader to \cite{ARS_ISTC16} for the details of density evolution and focus in this paper on 
the error floor due to the occurrence of a random burst of length $b<\beta_{\rm max} M$ when $M$ is finite.

\subsection{Burst-Erasure Channel Models}
\label{sec:SPBC}

We consider four different channel models
introducing either a single burst of erasures or few random bursts of erasures during the transmission of a SC-LDPC codeword.
First, the \emph{Single Position Burst Channel} (SPBC) erases all $M$ VNs of exactly one SP $z \in [L]$ in the transmitted codeword, so that bit indices $E = \{(z-1)M+1,\ldots,z M\}$ are erased in the received word, while all other bits are received correctly. We additionally introduce random independent erasures to this model in Section~\ref{sec:floorBEC}.

The third model is the more general \emph{Random Burst Channel} (RBC) whose burst pattern is denoted by RBC($b$) so that bit indices $E = \{(z-1)M+S,\ldots,(z-1)M+S+b-1\}$ are erased with a random $S \in [M]$ with $\pr{S=i}=1/M$, $z\in [L]$, and $b<\beta_{\rm max}M $.

The fourth and final channel model is the \emph{Block Erasure Channel}, denoted by BLEC($p$), which completely erases each SP independently with probability $p$.
This channel model includes the case of multiple independent bursts, which can occur in some practical scenarios like distributed storage. We show how the basic single-burst models can be used as building structures to closely estimate the decoding failure probability in 
the BLEC($p$). A similar approach may be applicable for other more complex channel models.

The simple single-burst channel model has for instance been used in~\cite{Iyengar-icc10} to study the recoverability of single bursts in protograph-based SC-LDPC codes. While multiple other models exist for a correlated erasure channel, like the Gilbert-Elliott model~\cite{Iyengar-icc10}, 
the above models can also describe some realistic scenarios:
for instance, the BLEC can model simultaneous multiple node failures in a distributed storage scenario with nodes associated to SPs. 
Or, the SPBC can be used to model a slotted-ALOHA multiple access scheme where each user transmits an SC-LDPC codeword over $L$ time slots, but one SP might be erased in the case of a single collision event. 
Moreover, long burst erasures might occur in block fading scenarios, in optical communications which are subject to polarization dependent loss. 

\section{Error Analysis on the SPBC}
\label{sec:spbcAnalysis}

\newcommand{\Nsp}{\mathbb{N}_2^{\rm\scriptscriptstyle SP}}
\newcommand{\Nspone}{\mathbb{N}_1^{\rm\scriptscriptstyle SP}}

As highlighted above, we empirically observed in~\cite{ARS_ISTC16} that the worst burst scenario is when the starting position of the burst is the first bit of a spatial position in the codeword.
Hence, we first analyze the performance of SC-LDPC code ensembles on the SPBC in this section and later generalize to the RBC. 
We give lower bounds on the block erasure rate after decoding and also provide a comparison with the Poisson ensemble.

First, we start by giving a necessary condition for being able to correct a burst.

\begin{proposition}\label{prop:boundebp}
Consider the $\mathcal{C}_{\mathcal{R}}(d_v,d_c,w,L,M)$ SC-LDPC ensemble and transmission over the SPBC, where exactly one SP is erased and all other SPs are received correctly. 
A necessary condition for recovering the erased spatial position is that $w\geq \left\lceil 1/\eBP(d_v,d_c)\right\rceil$, where $\eBP(d_v,d_c)$ is the BP threshold of the underlying ($d_v,d_c$) LDPC ensemble.
\end{proposition}
\begin{proof}
Consider the $\mathcal{C}_{\mathcal{R}}(d_v,d_c,w,L,M)$ ensemble, where we only consider erasures of entire spatial positions (not bits), and so $M$ is immaterial. 
The DE equation of this ensemble is given by~\cite{Kudekar-it11}
\[
x_i^{(\ell)} = \epsilon_i\left(1 - \frac{1}{w}\sum_{j=0}^{w-1}\left(1 -\frac{1}{w}\sum_{k=0}^{w-1}x_{i+j-k}^{(\ell-1)}\right)^{d_c-1}\right)^{d_v-1} ,
\]
where $\epsilon_i$ denotes the erasure probability at SP $i$. For the SPBC, we have $\epsilon_z=1$ and $\epsilon_{\sim z} = 0$, where the subscript $\sim z$ denotes all SPs $i\in[L]\backslash\{z\}$. 
We can immediately see that $x_{\sim z}^{(\ell)}=0$ and then we can simplify the DE equation as
\begin{align*}
x_z^{(\ell)} &= \left(1 - \frac{1}{w}\sum_{j=0}^{w-1}\left(1 - \frac{1}{w}x_z^{(\ell-1)}\right)^{d_c-1}\right)^{d_v-1} \\
  &= \left(1 -\left(1 - \frac{1}{w}x_z^{(\ell-1)}\right)^{d_c-1}\right)^{d_v-1} .
\end{align*}
After a change of variable $\tilde{x}_z^{(\ell)} = x_z^{(\ell)}/w$, we get
\[
\tilde{x}_z^{(\ell)} = \frac1w\left(1 -\left(1 - \tilde{x}_z^{(\ell-1)}\right)^{d_c-1}\right)^{d_v-1},
\]
which resembles the well-known DE equation for regular $(d_v,d_c)$ LDPC ensembles. 
We see that we can recover the missing SP if and only if $1/w\leq \eBP(d_v,d_c)$, where $\eBP(d_v,d_c)$ is the BP threshold of the underlying regular uncoupled ensemble~\cite{Richardson-MCT08}. 
Rearranging this condition completes the proof.
\end{proof}
\begin{corollary}
Consider the $\mathcal{C}_{\mathcal{R}}(d_v,d_c,w,L,M)$ SC-LDPC ensemble and transmission over the SPBC. 
A relaxed necessary condition for recovering the erased spatial position is that $w\geq \left\lceil \frac{d_c}{d_v}\right\rceil$.
\end{corollary}
\begin{proof}
By using the fact that $\eBP(d_v,d_c) \leq 1-R \leq \frac{d_v}{d_c}$, the relaxed condition is immediately obtained from Proposition~\ref{prop:boundebp}.
\end{proof}

\begin{remark}\label{rem:mult_SPs}
Note that this result is not only valid for a single burst, but in general for any combination of erased SPs out of the $L$ total SPs, provided that they are at least $w$ apart, i.e., there are at least $w-1$ non-erased SPs between two erased SPs. 
\end{remark}

\begin{remark} The analysis based on density evolution, e.g., Proposition~\ref{prop:boundebp},
estimates the ``bit error probability'' in the limit of $M$ while our main concern is ``block error probability'' when there is a burst of erasures. It has been proven for LDPC ensembles with $d_v\geq 3$ that a vanishing bit error probability guarantees a vanishing block error probability over the BEC~\cite[Lemma 3.166]{Richardson-MCT08}. Although it might not be true in general, our simulations suggest the same behaviour for SC-LDPC codes
when $d_v\geq 3$ and the channel is affected by a burst of erasures.   
\end{remark}

\subsection{The Random Regular SC-LDPC Ensemble ($\mathcal{C}_\mathcal{R}$)}
\label{sec:C_R}

\begin{sloppypar}
Let $P_\B^\SPBC$ denote the average block erasure (decoding failure) probability of the random $\mathcal{C}_{\mathcal{R}}(d_v,d_c,w,L,M)$ ensemble on the SPBC under BP decoding, i.e., the probability that the iterative decoder fails to recover the codeword. 
For large enough $M$, size-$2$ stopping sets (each of which also forms a weight-2 codeword) are the dominant structures in the graph that cause the BP decoder to fail~\cite{Olmos-isit11}. 
Hence, the number of size-$2$ stopping sets per SP, denoted $\Nsp$, is a good starting point for analyzing the performance of the ensemble on the SPBC.
We introduce the stopping set indicator function $U_{ij}$ with
\[
U_{ij} = \left\{\begin{array}{ll}
1 &, \text{if VNs }v_i\text{ and }v_j\text{ form a stopping set}\\
0 &, \text{otherwise}
\end{array}\right. .
\]
We clearly have $\mathbb{E}[U_{ij}] = \pr{U_{ij}=1}$. Thus, $\Nsp=\sum_{1\leq i<j\leq M} U_{ij}$ where the summation is over all $\binom{M}{2}$ pairs of VNs $v_i$ and $v_j$ from an SP.
\end{sloppypar}

\begin{lemma}
\label{lem:PU_indic}
Consider a code sampled uniformly from the $\mathcal{C}_{\mathcal{R}}(d_v,d_c,w,L,M)$ SC-LDPC ensemble.
The probability that two variable nodes from a same SP $z$ of this code form a stopping set amounts to
\begin{align}
P_{\mathcal{R}} = 
\frac{\left(1-\frac1{d_c}\right)^{d_v}}{\sum\limits_{\ell=0}^{d_v} \binom{d_v}{\ell}
\binom{wM\frac{d_v}{d_c}-d_v}{d_v-\ell} \left(1-\frac1{d_c}\right)^\ell }. 
\label{eq:lem1_pdef}
\end{align}
\end{lemma}
\begin{proof}
Let $v_i$ and $v_j$ be two VNs randomly chosen from an SP $z$ of the $\mathcal{C}_{\mathcal{R}}(d_v,d_v,w,L,M)$ ensemble, i.e., $i,j \in \{(z-1)M+1,\ldots,zM\}$ for $z \in [L]$. 
Recall that this ensemble contains no parallel edges. 
We use a combinatorial argument to compute the probability $P_{\mathcal{R}}$ that these two VNs form a size-$2$ stopping set. 
We label all the sockets of CNs. 
Let $T$ denote the total number of possible sub-graphs from $\{v_i,v_j\}$ and let $T_{ss}$ denote the number of possible sub-graphs in which these VNs form a size-$2$ stopping set.
First, we connect the $d_v$ edges of $v_i$ to randomly chosen empty sockets of $d_v$ distinct CNs as described in Section~\ref{sec:randomSCLDPC}. 
A stopping set (and in this case, also a low-weight codeword) is formed if and only if the edges of $v_j$ are connected to the same CNs, i.e., $\mathcal{N}(v_i)$. 
Each of these CNs has $d_c-1$ free distinct sockets. 
Thus,
\[
T_{ss}  = d_v! (d_c-1)^{d_v}, 
\]
where $d_v!$ is due to the permutation of edges and $(d_c-1)^{d_v}$ is due to the different ways of 
connecting to free sockets of $\mathcal{N}(v_i)$. The counting argument is illustrated by an example in Fig.~\ref{fig:size2ss}. 
\begin{figure}[t!]
\centering
\myfigurenoresize{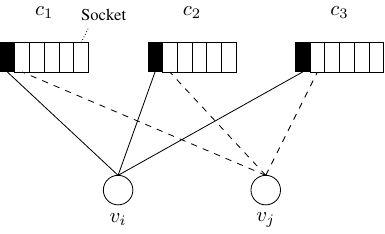}
\caption{A size-$2$ stopping set in a code from the $(3,6)$ LDPC random ensemble. CNs $\{c_1,c_2,c_3\}$ and VNs $\{v_i,v_j\}$ have been labeled for convenience. CNs have been expanded to show all their $d_c=6$ sockets. The solid edges indicate definite connections and the dashed edges complete to form a stopping set. Parallel edges are not allowed in the ensemble.}\label{fig:size2ss}
\end{figure}
In general, provided that $v_i$ connects to some $d_v$ check nodes, $v_i$ and $v_j$ may connect to some $\ell$ common CNs, $0 \leq \ell\leq d_v$. 
On the one hand, there are $\binom{d_v}{\ell}(d_c-1)^\ell$ socket selections for the $\ell$ common CNs. 
On the other hand, there are $\binom{wM\frac{d_v}{d_c}-d_v}{d_v-\ell}d_c^{d_v-\ell}$ socket selections for all other distinct $wM\frac{d_v}{d_c}-d_v$ CNs. 
Including $d_v!$ permutation of edges, we have,
\[
T=d_v!\sum\limits_{\ell=0}^{d_v} \binom{d_v}{\ell}\binom{wM\frac{d_v}{d_c}-d_v}{d_v-\ell} (d_c-1)^\ell (d_c)^{d_v-\ell}.
\]
We get $P_{\mathcal{R}} \triangleq \pr{U_{ij}=1} = \frac{T_{ss}}{T}$, simplified further to \eqref{eq:lem1_pdef}.
\end{proof}
\begin{remark}
For large enough $M$, the quantity $T$ in the proof of Lemma~\ref{lem:PU_indic} can be well approximated by the dominating summand ($\ell=0$) which leads to the following approximations
 \begin{align}
\pr{U_{ij}=1} &\approx \frac{(1-1/d_c)^{d_v}}{\binom{wM\frac{d_v}{d_c} -d_v}{d_v}}
\approx d_v!\left(\frac{d_c-1}{(wM-d_c)d_v}\right)^{d_v}. \label{eq:p2approx}
\end{align}
\end{remark}
Since $\pr{U_{ij}=1}$ is identical for all pairs $(i,j)$ of VNs, we have $\mathbb{E}[\Nsp]=\binom{M}{2}\pr{U_{ij}=1}$ implying the following result.

\begin{theorem}\label{thm:thmlb}
Consider a code sampled uniformly from the $\mathcal{C}_{\mathcal{R}}(d_v,d_c,w,L,M)$ ensemble with 
$wM\geq 2(d_v+1)d_c$. 
If all variable nodes of a randomly chosen SP are erased, 
the (average) probability of BP decoding failure is lower-bounded by
\begin{align}
P_\B^\SPBC \geq \binom{M}{2}\left(1 - \frac{M^2}{(\frac{w}{d_c}M-3)^{d_v}}\right)
P_{\mathcal{R}},
\label{eq:lower_bound}
\end{align}
where 
$P_{\mathcal{R}}$ is the probability that two variable nodes from an SP of the code form a stopping set, given by~\eqref{eq:lem1_pdef}.
\end{theorem}
\begin{proof}
See Appendix~\ref{sec:appProof1}.
\end{proof}
\begin{remark}
Since a size-2 stopping set characterizes a low-weight codeword, Theorem~\ref{thm:thmlb} yields
a lower-bound on the average block erasure probability of MAP decoding, too.

Furthermore, $P_\B^\SPBC$ is the average block erasure probability over the code ensemble. 
In the limit of $M$, the block erasure probability of each instance of the ensemble concentrates to its average $P_\B^\SPBC$ (see \cite[Theorem 3.30]{Richardson-MCT08}).
\end{remark}

Theorem~\ref{thm:thmlb} also implies that for large enough $M$,
\begin{align}
P_\B^\SPBC & \geq  \mathbb{E}[\Nsp] \biggr( 1-O\biggr(\frac{1}{M^{d_v-2}}\biggr) \biggr)\approx \mathbb{E}[\mathbb{N}_2^{\rm\scriptscriptstyle SP}] \triangleq \lambda_{\rm SP}.
\end{align}
Note that following standard arguments~\cite{Olmos-isit11}, \cite[Appendix C]{Richardson-MCT08}, for large $M$, we can also approximate the bound on $P_\B^\SPBC$ by modeling $\mathbb{N}_{2}^{\rm SP}$ as a Poisson distribution with mean $\lambda_{\rm SP}$, i.e.,
\begin{equation}
\label{eq:SPBCLowerBound}
P_\B^\SPBC \gtrapprox 1-\pr{\mathbb{N}_{2}^{\rm SP}=0} \approx 1-e^{-\lambda_{\rm SP}} \approx \lambda_{\rm SP}.
\end{equation}
We observe that the average block erasure probability scales as $O(M^{2-d_v})$.

\subsection{Simulation Examples}
\label{sec:SPBCSim}

\begin{figure}[tb!]
\centering
\myfigure{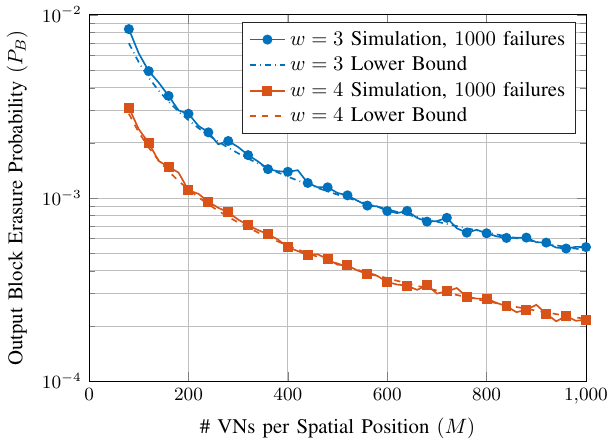}
\caption{Monte Carlo simulations over the SPBC with the $\mathcal{C}_{\mathcal{R}}(3,6,w,L,M)$ ensemble for $w=3$ and $w=4$, along with their respective lower bounds~(\ref{eq:lower_bound}).}\label{fig:SPBC_w_3_w_4}
\end{figure}

\begin{figure}[tb!]
\myfigure{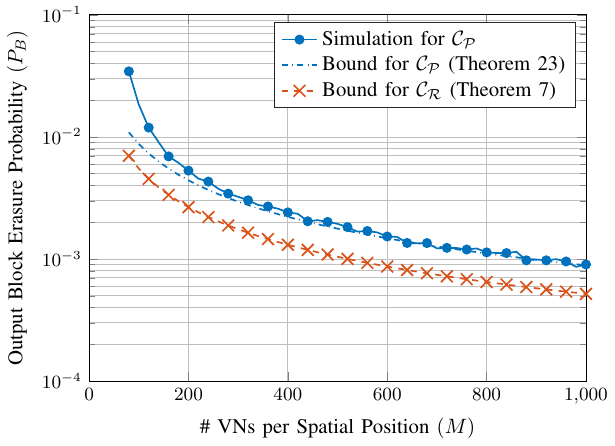}
\caption{Monte Carlo simulations over the SPBC with the $\mathcal{C}_{\mathcal{P}}(3,6,3,L,M)$ ensemble, along with the lower bounds according to Theorem~\ref{thm:thmlb_poisson}. The bound for the $\mathcal{C}_{\mathcal{R}}$ ensemble is also plotted for comparison.}\label{fig:SPBC_poisson_w_3}
\end{figure}

To illustrate the accuracy of our bounds, we carried out Monte-Carlo simulations where we randomly erased a complete spatial position (i.e., all $M$ VNs) from the middle of the graph (to avoid boundary effects) for each transmitted codeword. 
At the receiver we performed BP decoding and estimated the error rate $P_\B^\SPBC$ averaged over the ensemble by counting $1000$ decoding failures in each experiment.
The simulation results for $\mathcal{C}_{\mathcal{R}}(3,6,w,L=100,M)$ with $w=3$ and $w=4$ are shown in Fig.~\ref{fig:SPBC_w_3_w_4} along with their respective lower bounds calculated using~\eqref{eq:lower_bound} and \eqref{eq:p2approx}. 
We observe that the bound indeed is very accurate for $M$ large enough ($\gtrapprox 180$), since large-size stopping sets (larger than $2$) vanish. 
The simulation curve is slightly unstable because counting $1000$ failures is not enough to keep the sample variance small as $P_\B^\SPBC$ decreases by $O(M^{2-d_v})$.
By counting $1000$ failures, the simulation results are within $\pm 6.2\%$ of the theoretical value with a $95\%$ confidence interval.

We repeated the experiment for the $\mathcal{C}_{\mathcal{P}}(3,6,3,L,M)$ Poisson ensemble. Using the technique of the second moment method, we lower-bounded the block error probability of this ensemble on the SPBC. 
The lower-bound and some related results are presented in Appendix~\ref{sec:poisson}.
The simulation results are given in Fig.~\ref{fig:SPBC_poisson_w_3} along with the lower bound calculated using \eqref{eq:p2poisson} and \eqref{eq:lower_bound_poisson} in Appendix~\ref{sec:poisson}. 
For comparison, we have also plotted the lower bound of Theorem~\ref{thm:thmlb} for a $\mathcal{C}_{\mathcal{R}}(3,6,3,L,M)$ random ensemble. 
As proved in Theorem~\ref{thm:SC_Poisson_Comp} of Appendix~\ref{sec:poisson}, we observe that the Poisson ensemble always has a larger probability of size-2 stopping sets than the corresponding random ensemble. 
Since the ensemble $\mathcal{C}_{\mathcal{P}}$ is hence shown to be suboptimal compared to $\mathcal{C}_{\mathcal{R}}$, we do not consider this ensemble in further discussions.
We discussed the simple non-expurgated ensembles here mainly to illustrate the proof technique based on the second moment method and to show that simple lower bounds on the performance exist. The removal of short cycles in the graph, as it is frequently done in practical code constructions, is considered in the next section.

\section{Effects of Expurgation on the SPBC}
\label{sec:expurgate}

\subsection{Minimal Stopping Set Size}
\label{sec:minimalSS}

As the performance of SC-LDPC codes over (burst and random) erasure channels is mainly dominated by size-$2$ stopping sets, it is well known that we can improve the burst erasure correction capability by expurgating the ensemble and thereby removing all small stopping sets. 
Observing that a size-$2$ stopping set, as shown in Fig.~\ref{fig:size2ss}, is built around 4-cycles, we can reduce the size of the minimal stopping sets by removing small cycles from the graph. 
For example, increasing the girth of the graph to $6$ leads to minimal stopping sets of size $s_{\rm min} = d_v+1$~\cite{Orlitsky-isit02}.

\begin{figure}[t!]
\centering
\includegraphics[width=0.85\columnwidth,keepaspectratio]{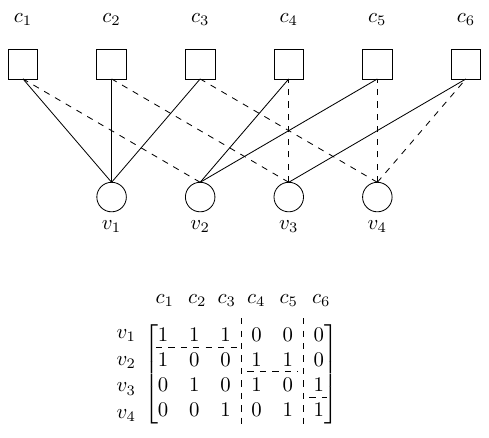}
\caption{A size-$4$ stopping set in a code from the expurgated $(3,6,w,L,M)$ random ensemble. CNs $\{c_1,c_2,c_3,c_4,c_5,c_6\}$ and VNs $\{v_1,v_2,v_3,v_4\}$ have been labeled for convenience. The solid edges indicate definite connections and the dashed edges complete one configuration to form a stopping set. Parallel edges are not allowed in the ensemble. The transposed bi-adjacency matrix is also shown with its pattern highlighted.}\label{fig:size4ss}
\end{figure}

In this section, we first show that all size $(d_v+1)$ stopping sets in an LDPC code share a common structure which we will use to find the probability of these stopping sets.
Consider a $(3,6)$ random LDPC code ensemble as an example. 
We immediately notice that size-$3$ stopping sets vanish for girth-6 graphs.
A size-$4$ stopping set is shown in Fig.~\ref{fig:size4ss} along with its (transposed) bi-adjacency matrix that describes the neighbors of each VN in the corresponding row. 
We can notice a pattern in this matrix which has been highlighted using dashed lines in the matrix: row $i \in \{ 1,2,\ldots,d_v \}$ has one subset of $(d_v-(i-1))$ columns with all $1$s and an identity matrix $I_{d_v-i+1}$ spanning these columns starting from row $i+1$. 
In the following lemma, we characterize this pattern and relate it to smallest size stopping sets (SSs) in graphs of girth 6.
\begin{lemma}\label{lem:girth6SS}
In a regular bipartite graph of girth $6$ (i.e., without parallel edges), the smallest stopping set is  of size $d_v+1$, involves exactly $\frac{1}{2}d_v(d_v+1)$ neighboring CNs and is characterized by a bi-adjacency matrix that contains as rows all $\binom{d_v+1}{2} = \frac{1}{2}d_v(d_v+1)$ permutations of the binary vector $\begin{pmatrix} 1 & 1 & 0 & \cdots & 0\end{pmatrix}$ of length $d_v+1$ containing two ``1''s and $d_v-1$ ``0''s.
\end{lemma}
\begin{proof}
From~\cite{Orlitsky-isit02}, we know that the smallest stopping set is of size $d_v+1$. 
Furthermore the stopping set connects to exactly $\frac{1}{2}(d_v+1)d_v$ check nodes. 
This can be seen as follows: the number of neighboring check nodes cannot be larger than $\frac{1}{2}(d_v+1)d_v$ as the total number of $(d_v+1)d_v$ edges must connect to every check node at least twice, which is only fulfilled if the number of check nodes is at most $\frac{1}{2}(d_v+1)d_v$. Conversely, if there are less than $\frac{1}{2}(d_v+1)d_v$ neighboring check nodes, by the pigeonhole principle, at least one check node connects to more than two neighboring variable nodes in the SS. 
Assume w.l.o.g. that there is a single CN, called $c_1$, with 3 neighbors $v_1$ $v_2$ and $v_3$ in the SS. 
The VN $v_1$ must connect to $d_v-1$ CNs $c_2,\ldots,c_{d_v}$ other than $c_1$ as parallel edges are avoided. These $d_v-1$ CNs in turn 
need to connect to $d_v-1$ \emph{distinct} VNs other than $v_1$, $v_2$ and $v_3$ to avoid a cycle of length 4. 
Hence, we need $3+d_v-1 = d_v+2$ VNs which contradicts the assumption that the SS consists of $d_v+1$ VNs.

From the previous considerations, all CNs must connect exactly twice to the variable nodes and $|\mathcal{N}(v_i)\cap\mathcal{N}(v_j)| < 2$, for all $i\neq j$ as cycles of length-4 are avoided. 
Thus, the bi-adjacency matrix describing the SS consists of rows, representing connections of CNs, of weight $2$ and no row can be used twice as this violates the previous assumption of avoiding cycles of length $4$. 
The fact that the SS consists of $\frac{1}{2}d_v(d_v+1)$ CNs and that there are only $\binom{d_v+1}{2} = \frac{1}{2}d_v(d_v+1)$ distinct binary vectors of length $d_v+1$ and weight 2 proves the lemma.
\end{proof}

\subsection{Estimation of Block Error Rate}
\label{sec:no4SPBCAnalysis}

We can use the same approach as in Lemma~\ref{lem:PU_indic} to calculate the probability of occurrence of the minimum size stopping sets, characterized by Lemma~\ref{lem:girth6SS}, within a spatial position of a code sampled uniformly from the ensemble.

\begin{lemma}\label{lem:girth6_PR}
For a code sampled uniformly from the expurgated $\mathcal{C_{R}}(d_v,d_c,w,L,M)$ SC-LDPC ensemble, constructed without allowing cycles of length $4$, the probability $P_{\mathcal{R},6}$ that $d_v+1$ variable nodes of the same spatial position $z$ form a stopping set is bounded by
\begin{align}
K_6 & \binom{wM\frac{d_v}{d_c}}{\frac{1}{2}d_v(d_v+1)} \binom{wM\frac{d_v}{d_c}}{d_v}^{-(d_v+1)} \nonumber \\ 
& \leq P_{\mathcal{R},6} \leq 
\tilde{K}_6 \binom{wM\frac{d_v}{d_c}-\frac{1}{2}d_v(d_v+1)}{\frac{1}{2}d_v(d_v+1)}^{-1} \label{eq:no4pSS1}\\
& \qquad \quad \ \leq K_6 \binom{wM\frac{d_v}{d_c}}{\frac{1}{2}d_v(d_v+1)}\binom{wM\frac{d_v}{d_c}-d_v^2}{d_v}^{-(d_v+1)} , \label{eq:no4pSS1_bad}
\end{align}
with 
\[
K_6 = \left(1-\frac{1}{d_c}\right)^{\frac{1}{2}d_v(d_v+1)}(\frac{1}{2}d_v(d_v+1))! 
\]
and
\[
\tilde{K}_6 = \left(1-\frac{1}{d_c}\right)^{\frac{1}{2}d_v(d_v+1)} \frac{(d_v!)^{d_v+1}}{(\frac{1}{2}d_v(d_v+1))!}
\]
being constants depending only on the ensemble parameters $d_v$ and $d_c$.
\end{lemma}
\begin{proof}
Let $v_{i_1},\ldots, v_{i_{d_v+1}}$ denote $d_v+1$ randomly chosen VNs from an SP $z$ of a code sampled from the expurgated, girth-$6$, $\mathcal{C_{R}}(d_v,d_c,w,L,M)$ SC-LDPC ensemble. 
Recall that this code contains neither parallel edges nor cycles of length $4$. 
Let $T$ denote the total number of possible sub-graphs from $\{v_{i_1},\ldots, v_{i_{d_v+1}}\}$ and let $T_{ss}$ denote the number of sub-graphs in which these VNs form a stopping set.

We first compute $T_{ss}$. Consider the bi-adjacency matrix of a stopping set of size $d_v+1$ described in Lemma~\ref{lem:girth6SS}. and label the rows and column as shown in Fig.~\ref{fig:size4ss}. Any column-wise permutation of this matrix represents the bi-adjacency matrix of another stopping set with the same VNs, and vice versa. Thus, the number of such sub-graphs is equal to the number of distinct ``labeled'' permutations times the permutation of sockets of each CNs. In other words, 
\begin{equation*}
T_{ss} = \binom{wM\frac{d_v}{d_c}}{\frac{1}{2}d_v(d_v+1)}  (\frac{1}{2}d_v(d_v+1))!\left(d_c(d_c-1)\right)^{\frac{1}{2}d_v(d_v+1)},
\end{equation*} 
where the first term is the number of possibilities to select $\frac{1}{2}d_v(d_v+1)$ CNs. The second term $(\frac{1}{2}d_v(d_v+1))!$ counts the number of column-wise permutations of the bi-adjacency matrix. The last term is the number of ways that the edges are connected to two distinct sockets of each CN. 

Unfortunately, obtaining an exact expression for $T$ is tedious and does not lead to many new insights. We therefore give upper and lower bounds for $T$ which become tight for large $M$.

First, $T$, the total number of sub-graphs from the $d_v+1$ VNs such that 4-cycles and parallel edges are avoided is upper-bounded by the total number of sub-graphs including 4-cycles. 
Hence, for every VN, we (randomly) select $d_v$ CNs out of $wM\frac{d_v}{d_c}$ CNs. 
Now, assume that each edge of the VN is paired to one of the chosen $d_v$ CNs.
For every selected CN, there are at most $d_c$ possibilities to select a free socket to place the paired edge from the VN. 
Hence, we get the bound
\begin{align}
T \leq \left[\binom{wM\frac{d_v}{d_c}}{d_v}d_c^{d_v}\right]^{d_v+1}.\label{eq:ss4_lb}
\end{align}
We now proceed to the lower bound. Given a set of $(d_v+1)$ randomly chosen VNs $v_{i_1},\ldots, v_{i_{d_v+1}}$ from a spatial position, the VN $v_{i_1}$ has to connect its $d_v$ edges to CN sockets without introducing parallel edges or $4$-cycles.
So, the first VN has $wMd_v/d_c$ CNs to chose from. However, once all edges of $v_{i_1}$ are connected, the first edge of $v_{i_2}$ can either connect to one of the $d_v$ previously connected CNs or to one of the $(wM-d_c)d_v/d_c$ previously unconnected CNs. In the former case, none of the remaining edges can connect to another previously connected CN as a $4$-cycle would be formed. As the latter case is dominant, we lower bound $T$ by only considering that case where all edges connect to previously unconnected CNs. We thus get, following a similar line of reasoning as before,
\begin{align}
T \geq \left[d_c^{d_v}\right]^{d_v+1}\prod_{j=0}^{d_v}\binom{wM\frac{d_v}{d_c}-jd_v}{d_v}\label{eq:intermediate_ss4}
\end{align}
which we can further bound as
\begin{align}
T \geq \left[\binom{wM\frac{d_v}{d_c}-d_v^2}{d_v}d_c^{d_v}\right]^{d_v+1}.\label{eq:ss4_ub1}
\end{align}
We can also simplify~\eqref{eq:intermediate_ss4} using the fact that
\begin{align*}
\prod_{j=0}^{d_v} & \binom{wM\frac{d_v}{d_c}-jd_v}{d_v} \\
 &= \frac{(wM\frac{d_v}{d_c})!}{(wM\frac{d_v}{d_c}-d_v(d_v+1))!(d_v!)^{d_v+1}} \\
 &= \binom{wM\frac{d_v}{d_c}}{\frac{1}{2}d_v(d_v+1)}\binom{wM\frac{d_v}{d_c}-\frac{1}{2}d_v(d_v+1)}{\frac{1}{2}d_v(d_v+1)} \\
 & \hspace{4.5cm} \times\frac{([\frac{1}{2}d_v(d_v+1)]!)^2}{(d_v!)^{d_v+1}}
\end{align*}
which leads to
\begin{align}
T & \geq d_c^{d_v(d_v+1)}\binom{wM\frac{d_v}{d_c}}{\frac{1}{2}d_v(d_v+1)}\binom{wM\frac{d_v}{d_c}-\frac{1}{2}d_v(d_v+1)}{\frac{1}{2}d_v(d_v+1)} \nonumber \\
 & \hspace{4.75cm} \times\frac{([\frac{1}{2}d_v(d_v+1)]!)^2}{(d_v!)^{d_v+1}} . \label{eq:ss4_ub2}
\end{align}
Finally, simplification of $P_{\mathcal{R},6} = \frac{T_{ss}}{T}$ with the bounds~\eqref{eq:ss4_lb}, \eqref{eq:ss4_ub1}, and~\eqref{eq:ss4_ub2} leads to the result.  
Note that in the counting argument we skipped all $(d_v+1)!(d_v!)^{d_v+1}$ permutations of VNs and their sockets as they are included in both $T_{ss}$ and $T$ and henceforth cancel in the final expression.
\end{proof}

\begin{remark}
For a code sampled uniformly from the expurgated $\mathcal{C}_{\mathcal{R}}(d_v,d_c,w,L,M)$ ensemble with girth $6$,  we can approximate $T$ quite well by~\eqref{eq:ss4_ub2}, which is the dominant term in the expression of $T$ for large enough $M$ and get the approximation
\begin{align}
P_{\mathcal{R},6} \approx \tilde{K}_6 \binom{wM\frac{d_v}{d_c}-\frac{1}{2}d_v(d_v+1)}{\frac{1}{2}d_v(d_v+1)}^{-1}.\label{eq:pr6_approx}
\end{align}
\end{remark}

\begin{figure}[tbh!]
\myfigure{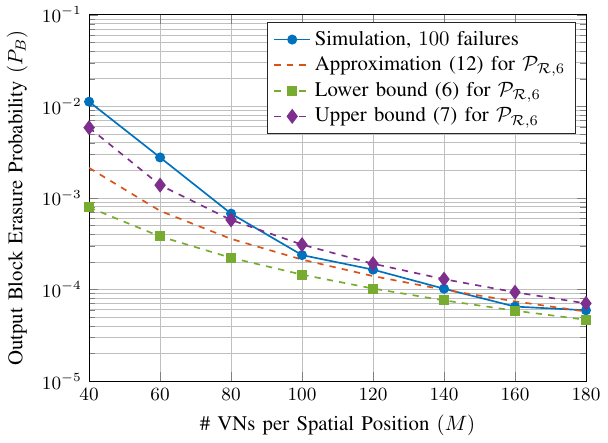}
\caption{Monte Carlo simulations over the SPBC with the expurgated $\mathcal{C}_{\mathcal{R}}(3,6,3,L,M)$ ensemble along with the theoretical bounds and approximations.}\label{fig:SPBC_no4cycles}
\end{figure}

\begin{figure}[tbh!]
\myfigure{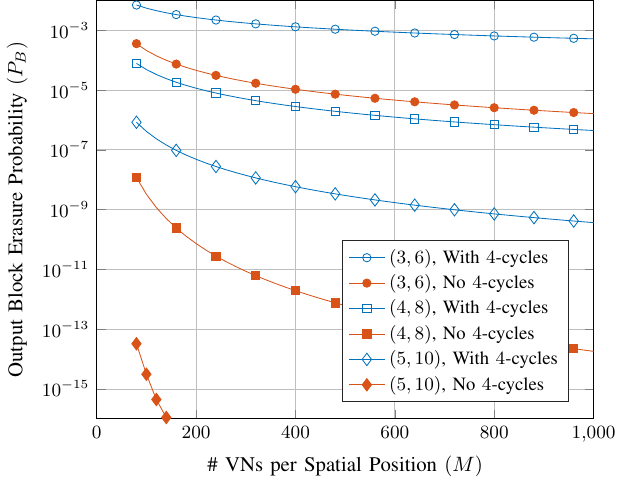}
\caption{Approximations on $P_{B}^\SPBC$  for $\mathcal{C}_{\mathcal{R}}(d_v,d_c,w=d_v,L,M)$ ensembles with $d_v=3$, $d_v=4$, and $d_v=5$ in both the unexpurgated and expurgated (using~\eqref{eq:pr6_approx}) scenarios.}\label{fig:effect_of_no4cycles}
\end{figure}

We performed Monte Carlo simulations for the expurgated $\mathcal{C}_{\mathcal{R}}(3,6,3,L,M)$ ensemble and counted $100$ decoding failures on the SPBC. Figure~\ref{fig:SPBC_no4cycles} illustrates
the simulation averages for the block error probability, for varying $M$, their respective lower and upper bounds (using the bounds of $\mathcal{P}_{\mathcal{R},6}$ from Lemma~\ref{lem:girth6_PR}), and the approximation using~\eqref{eq:pr6_approx}. 
It is evident that both the bounds and the approximation become tight very quickly, which reassures that the decoder performance is indeed dominated by minimal stopping sets.

\begin{corollary}
The expected number of such stopping sets within a SP of the code is $\lambda_{SP} \triangleq \binom{M}{d_v+1} P_{\mathcal{R},6}$. 
Using similar arguments as in Section~\ref{sec:spbcAnalysis}, we have 
\begin{equation*}
\mathbb{N}_{d_v+1}^{SP} \sim \text{Poisson}(\lambda_{SP}) .
\end{equation*}
A tight approximation for the average block erasure probability on the SPBC, $P_{\B,6}^\SPBC$, is then obtained as
\begin{equation}
\label{eq:no4SPBCApprox}
P_{\B,6}^\SPBC \approx 1-e^{-\lambda_{\rm SP}}\approx \lambda_{\rm SP}.
\end{equation} 
\end{corollary}

We can see from the bounds~\eqref{eq:no4pSS1} and~\eqref{eq:no4pSS1_bad} that the average block erasure probability in this case scales as $O(M^{-d_v^2/2 + d_v/2 + 1})$, or equivalently $O(M^{-\frac{1}{2}(d_v-2)(d_v+1)})$, which is in contrast to the non-expurgated case where the scaling behavior was $O(M^{-(d_v-2)})$. Hence, unsurprisingly, expurgating always improves the average block erasure probability in the SPBC case.
These scaling behaviours are illustrated in Fig.~\ref{fig:effect_of_no4cycles} for three $\mathcal{C}_{\mathcal{R}}(d_v,d_c,w=d_v,L,M)$ ensembles with $d_v=3$, $d_v=4$, and $d_v=5$. All three ensembles have the asymptotic design rate of $\frac{1}{2}$. 
We observe that:
\begin{itemize}
\item For the unexpurgated case, linearly increasing $d_v$ reduces the block error probability
by multiples of $1/M$.
\item When the ensemble is expurgated so that $girth = 6$,  unit increase in $d_v$ improves the performance by a factor of about $M^{-d_v}$.
\item  On the SPBC, block erasure rates less than $10^{-15}$ can be obtained by codes from an expurgated $\mathcal{C}_{\mathcal{R}}(5,10,5,L,M)$ ensemble, with $M$ as low as $120$. This encourages the use of short SC-LDPC codes in applications where single SPs may be erased (e.g., distributed storage of short files) and high reliability is required.
\end{itemize}

\section{Error Analysis on More General Burst Channel Models}
\label{sec:general_channels}
In this section, we study some other transmission scenarios with less restrictions on the channel model.
First, we consider the random burst channel (RBC) model 
in which the burst has a variable length and not necessarily erases complete SPs. 
Then, we consider the scenario where there is random background noise in addition to the burst erasure. 
For simplicity, we assume a binary erasure channel for the background noise and the SPBC model for the burst.
We present an extension of the analysis in previous sections to estimate the decoding failure probability in these scenarios. 
Finally, we consider the scenario of the block erasure channel (BLEC) in which each SP may be erased randomly with some probability. 
We show how such a channel model can be analyzed using the simple single-burst channel models, e.g. SPBC and RBC.

Although the expurgated ensembles are of more interest,
we present the error analysis for the $\mathcal{C}_{\mathcal{R}}(d_v,d_c,w,L,M)$ SC-LDPC ensemble with size-2 stopping sets because of two reasons. 
First, expurgation leads to very small decoding failure probabilities and it will be more challenging to illustrate the tightness of the error estimations in simulation. 
Second, a similar analysis can be conducted for the expurgated ensemble as we have shown for the SPBC in Section~\ref{sec:expurgate}, but it will be more involved and less insightful.

\subsection{Random Burst Channel (RBC)}
\label{sec:RBC}

We have seen in the previous section that the performance of SC-LDPC ensembles is asymptotically well characterized by stopping sets of size $2$.
We will now generalize this analysis to the case where there is a single burst of length $> M$ and with a random starting location within the codeword. 
Such a burst can span over multiple SPs because of its random starting position. 
In this case, stopping sets formed across coupled SPs will also contribute to decoding failures.
We will restrict ourselves to the case of a single burst. 
This case can then serve as a building block for the more complex case of having multiple  random bursts. 
For instance, the results can be easily generalized to the case with multiple bursts spaced apart by at least $w$ non-affected SPs.
We start by first giving a simple asymptotic condition on $w$ for burst erasure recovery. 
Although it is an asymptotically sufficient condition, we have observed in numerical experiments that it reflects the actual behavior quite well.
\begin{proposition}
Consider the $\mathcal{C}_{\mathcal{R}}(d_v,d_c,w,L,M)$ SC-LDPC ensemble and transmission over a channel where a random selection of $b$ consecutive VNs are erased and all other VNs are received without erasure. Let $\beta = \frac{b}{M}$ denote the normalized burst length. 
Furthermore, let $w \geq 1 + \lceil\beta\rceil$. 
Asymptotically (in the limit of $M$), the erased VNs can be recovered when $w\geq \left\lceil (1+\lceil\beta\rceil)/\eBP(d_v,d_c)\right\rceil$, where $\eBP(d_v,d_c)$ is the BP threshold of the underlying ($d_v,d_c$) uncoupled LDPC ensemble.
\end{proposition}
\begin{proof}
Consider the burst of normalized length $\beta=\frac{b}{M}$. This burst spans at most $1+\left\lceil\frac{b-1}{M}\right\rceil$ SPs, which is upper bounded by $c=1+\lceil\beta\rceil$. 
If a burst erasing $c$ consecutive SPs is recoverable, then a random burst of length $b$ is recoverable as well. 
Let $z$ denote the first SP affected by the random burst.
Using similar arguments as in the proof of Proposition~\ref{prop:boundebp}, we can write, $\forall \ i \in \{0,\ldots,c-1\}$,
\begin{align*}
x_{z+i}^{(\ell)} &= \varepsilon_{z+i}\left(1 - \frac{1}{w}\sum_{j=0}^{w-1}\left(1 - \frac{1}{w}\sum_{k=0}^{w-1}x_{z+i+j-k}^{(\ell-1)}\right)^{d_c-1}\right)^{d_v-1}\\
&\overset{(a)}{\leq}\varepsilon_{z+i} \left(1 - \left(1 - \frac{1}{w}\sum_{k=0}^{c-1}x_{z+k}^{(\ell-1)}\right)^{d_c-1}\right)^{d_v-1}\\
&\overset{(b)}{\leq} \left(1 - \left(1 - \frac{1}{w}\sum_{k=0}^{c-1}x_{z+k}^{(\ell-1)}\right)^{d_c-1}\right)^{d_v-1},
\end{align*}
where $\varepsilon_{z+i}$ is the fraction of VNs erased in SPs $z+i$. 
We have $(a)$ as $x_z^{(l)}$ is nonzero only for $\{z,\ldots,z+c-1\}$, and thus, for any $i$ and $j$,
\begin{equation*}
\sum_{k=0}^{w-1}x_{z+i+j-k}^{(\ell-1)}\leq \sum_{k=0}^{c-1}x_{z+k}^{(\ell-1)}.
\end{equation*}
We have $(b)$ as $\varepsilon_{z+i}\leq 1$. The bound on $x_{z+i}^{(\ell)}$ is independent of $i$ and thus,
\begin{equation*}
\sum_{i=0}^{c-1}x_{z+i}^{(\ell)} \leq c\left(1 - \left(1 - \frac{1}{w}\sum_{k=0}^{c-1}x_{z+k}^{(\ell-1)}\right)^{d_c-1}\right)^{d_v-1} .
\end{equation*}
The substitution $y_z^{(\ell)} = \frac{1}{w}\sum_{k=0}^{c-1}x_{z+k}^{(\ell-1)}$ yields
\[
y_z^{(\ell)} \leq \frac{c}{w}\left(1 - \left(1 - y_z^{(\ell-1)}\right)^{d_c-1}\right)^{d_v-1} .
\]
If $c/w < \eBP(d_v,d_c)$, then $\lim_{\ell\to\infty} y_z^{(\ell)}=0$ and thus, 
$\lim_{\ell\to\infty} x_z^{(\ell)}=0$ for any $z$. It implies that the decoding failure (output bit error probability) converges to zero as well in the asymptotic limit of $M$ and $\ell$. 
\end{proof}
In the remainder of this section, we assume that asymptotic recovery is possible and we are concerned with decoding failures that occur when $M$ is finite.
Let $S\in[M]$ denote the random starting VN of the burst at some SP $z_0$. We define the normalized quantities $s=\frac{S}{M}$ and $\beta=\frac{b}{M}$. 
Recall from Section \ref{sec:betamax} that $\beta_{\rm max}$ is the largest $\beta$ asymptotically recoverable.
In the case of the RBC, besides stopping sets within a single spatial position, size-$2$ stopping sets formed across coupled SPs will also cause decoding failures. 
We first characterize the probability of such stopping sets in the following lemma. 

\begin{lemma}\label{lem:prob_rbc}
Consider the $C_{\mathcal{R}}(d_v,d_c,w,L,M)$ ensemble. Let $v_i$ denote a randomly chosen VN in spatial position 
$z\leq L$, and  $v_j$ denote a random VN in spatial position $z+k$, for a non-negative integer $k$ with $z+k\leq L$.
The probability that these two random VNs form a stopping set of size 2 is independent of $z$ and
amounts to
\begin{align}\label{eq:prob_rbc}
q_k=P_{\mathcal{R}}\left(1-\frac{k}{w}\right)^{d_v}, \quad k\in\{0,1,\ldots,w-1\},
\end{align}
where $P_{\mathcal{R}}$ is given by \eqref{eq:lem1_pdef}. For $k\geq w$, $q_k=0$.
\end{lemma}
\begin{proof}
From the definition of the ensemble, we know that $\mathcal{N}(v_i)$ is part of SPs $\{z,z+1,\dots,z+w-1\}$ and 
$\mathcal{N}(v_j)$ is part of SPs $\{z+k,z+k+1,\dots,z+k+w-1\}$.
A size-2 stopping set is formed if and only if $\mathcal{N}(v_i)=\mathcal{N}(v_j)$. 
For $k\geq w$, this condition cannot be fulfilled and thus, $q_k=0$.
For $k<w$, all check nodes of $\mathcal{N}(v_i)$ must be from a subset $\{z+k,\dots,z+w-1\}$. As the edges of the variable nodes uniformly connect to $w$ neighboring SPs, the probability of such a selection for $v_i$ is $(\frac{w-k}{w})^{d_v}$. 
Now, the probability that $v_j$ connects exactly to the same CNs as $v_i$, 
i.e., $\mathcal{N}(v_i)=\mathcal{N}(v_j)$, 
is equal to $P_{\mathcal{R}}$ (by the same argument as in Lemma~\ref{lem:PU_indic}). 
Hence, $q_k=\left(\frac{w-k}{w}\right)^{d_v}P_{\mathcal{R}}$. We immediately see that $q_0 > q_1 > \cdots > q_{w-1}$.
\end{proof}

The average number of size-2 stopping sets between VNs lying in SPs $z$ and $z+k$, where $k \in [w-1]$, is given by
\begin{equation}
\label{eq:lambdas}
\lambda_0 \triangleq \binom{M}{2} q_0 \qquad\text{and}\qquad \hspace{2.5mm} \lambda_k \triangleq M^2 q_k ,\quad k \in [w-1].
\end{equation}
Again, we see that $\lambda_k \sim O(M^{2-d_v})$.
To verify, consider the $C_{\mathcal{R}}(3,6,3,100,M=64)$ SC-LDPC ensemble. By averaging over all the SPs of $1000$ random code instances of the ensemble, the empirical average number of size-$2$ stopping sets is obtained as $(L\lambda_0,(L-1)\lambda_1,(L-2)\lambda_2)\approx(0.876,0.488,0.060)$, which is reasonably close to $(0.829,0.494,0.061)$ computed using~\eqref{eq:lambdas}, even for small $M=64$.

We now estimate the average decoding failure probability, $P_\B^{\RBC}$ when there is a burst of length $b\ll \beta_{\rm max}M$ with a random starting bit $M(z_0-1)+S$, $1\leq S\leq M$. For a given $(S,z_0,b)$, the number of erased VNs in SP $z$ is equal to
\begin{align}
\label{eq:m_z}
m_z &= \left\{\begin{array}{ll}
0 & z < z_0 \\
\min\{b,M-S+1\} & z = z_0 \\
\max\{0,\min\{M,b+S-1-(z-z_0)M\}\} & z>z_0.\end{array}\right.
\end{align}
Our error estimation is again based on the average number of size-2 stopping sets formed among erased VNs.
Let $\mathbb{N}_2(S,z_0,b)$ denote the number of size-2 stopping sets formed by VNs erased by the burst. BP decoding fails if $\mathbb{N}_2(S,z_0,b)\geq 1$. Thus,
\begin{equation}
P_\B^{\RBC}\geq \mathbb{P}\{\mathbb{N}_2(S,z_0,b)\geq 1\}\overset{(i)}{\approx} \mathbb{E}[\mathbb{N}_2(S,z_0,b)].
\label{eq:lowerbound}
\end{equation} 
There are two approaches to justify $(i)$. The first approach is to 
use arguments 
similar to Theorem~\ref{thm:thmlb} to lower-bound  $\mathbb{P}\{\mathbb{N}_2(S,z_0,b)\geq 1\}$ in terms of the average number of size-2 stopping sets and
a much smaller correction term. However, the derivation will be more 
involved than Theorem~\ref{thm:thmlb} and will not lead to new insights, which is why we omit it here.
An alternative is to use standard arguments~\cite[App. C]{Richardson-MCT08} to approximate the distribution of size-2 stopping sets by a
joint Poisson distribution. The decoding error then corresponds approximately to the average number of stopping sets.

\begin{theorem}\label{thm:thmapproxRBC}
Consider the $\mathcal{C}_{\mathcal{R}}(d_v,d_c,w,L,M)$ SC-LDPC ensemble affected by a burst of length $b=\beta M$ with a random starting bit  $M(z_0-1) + S$, where $1 \leq S \leq M$ and $ \beta\ll \beta_{\rm max}$.
The number of erased VNs in SP $z$ is given by $m_z$ in~\eqref{eq:m_z}.
Then the expected number of size-$2$ stopping sets formed by VNs erased by the burst is given by
\begin{align}
 & \mathbb{E}[\mathbb{N}_2(S,z_0,b)] \nonumber \\
 & \gtrapprox \frac{L-\lceil \beta\rceil}{(L-\beta)M+1}\sum_{S=1}^{M} \sum_{z=1}^{\lceil \beta\rceil+1}\left(\binom{m_z}{2}q_0+\sum_{k=1}^{w-1} m_z m_{z+k}q_k\right)\nonumber\\
&\approx
\frac{1}{M}\sum_{S=1}^{M} \sum_{z=1}^{\lceil \beta\rceil+1}\left(\binom{m_z}{2}q_0+\sum_{k=1}^{w-1} m_z m_{z+k}q_k\right).
\label{eq:errorfloor}
\end{align}
\begin{proof}
To find an expression for $\mathbb{E}[\mathbb{N}_2(S,z_0,b)]$, we first note that the starting position $M(z_0-1)+S$ of the burst is chosen uniformly among the bits in $[LM-b+1]$. Then,
\begin{align*}
& \mathbb{E}[\mathbb{N}_2(S,z_0,b)] \\
&= \frac{1}{LM-b+1}\sum_{k=1}^{\mathclap{LM-b+1}} 
\mathbb{E}[\mathbb{N}_2(S,z_0,b)\mid M(z_0-1)+S=k]\nonumber\\
&\overset{(i)}{\gtrapprox} \frac{M}{LM-b+1}\sum_{z_0=1}^{L-\lceil \beta\rceil}\frac{1}{M}\sum_{S=1}^{M} \mathbb{E}[\mathbb{N}_2(S,z_0,b)]\nonumber \\
&\overset{(ii)}{=}\frac{L-\lceil \beta\rceil}{(L-\beta)M+1} 
\sum_{S=1}^{M} \mathbb{E}[\mathbb{N}_2(S,z_0=1,b)]\nonumber\\
&\overset{(iii)}{=}\frac{L-\lceil \beta\rceil}{(L-\beta)M+1}\sum_{S=1}^{M} \sum_{z=1}^{\lceil \beta\rceil+1}\left(\binom{m_z}{2}q_0+\sum_{k=1}^{w-1} m_z m_{z+k}q_k\right),
\end{align*}
where $(i)$ is because we neglect a small positive contribution $(O(\frac{1}{L}))$ of starting positions larger than $(L-\lceil \beta\rceil)M$ for non-integer $\beta$. Recall that $\beta=b/M$. We have $(ii)$ as $\frac{1}{M}\sum_{S=1}^{M} \mathbb{E}[\mathbb{N}_2(S,z_0,b)]$ is identical for different $z_0$. Let us justify $(iii)$ using Lemma~\ref{lem:prob_rbc}: for a given starting bit $1\leq S\leq M$, 
$\binom{m_z}{2}q_0$ is the average number of size-2 stopping sets formed between erased VNs in SP $z$ and
$m_zm_{z+k}q_k$ is the average number of size-2 stopping sets formed between erased VNs in all pairs of SP $z$ and $z+k$, $k\in[w-1]$. We have $(iii)$ by summing the average number of size-2 stopping sets formed between erased VNs in all pairs of SP $z$ and $z+k$. 
\end{proof}
\end{theorem}

We verify Theorem~\ref{thm:thmapproxRBC} (and in particular~\eqref{eq:errorfloor}) for some specific choices of $b$ in the following examples.
\begin{example}[$b=M$]
\label{ex:b=M}
This is a generalized version of the SPBC where the starting bit of the burst is not constrained to occur at the exact boundary of the SP. In this case, the nonzero terms of \eqref{eq:errorfloor} are
$\binom{m_1}{2}q_0$, $\binom{m_2}{2}q_0$ and $m_1m_2q_1$. 
Then~\eqref{eq:lowerbound} has the following closed form
 \begin{align*}
P_\B^{\RBC}(M) \gtrapprox \frac{(M-1)(2M-1)}{6}q_0+ \frac{M^2-1}{6}q_1 ,
\end{align*}
which is obtained by summing all contributions after simplification of the sums. For large enough $M$, this term is well approximated by $\binom{M}{2}\frac{2q_0+q_1}{3}$. In contrast, the decoding error of the SPBC is approximately $P_\B^\SPBC\gtrapprox\binom{M}{2}q_0\gtrapprox P_\B^{\RBC}(M)$,  as $q_1 < q_0$.
\end{example}

\begin{example}[Tightness of \eqref{eq:errorfloor} for $b<\beta_{\rm max}M$] We plot the decoding failure probability of the $\mathcal{C_R}(3,6,w,L,M)$ ensemble for different finite values of $M$ and for $w=3,4$ in Fig.~\ref{fig:finite363} and Fig.~\ref{fig:finite364}.
The maximum normalized correctable burst length, $\beta_{\rm max}$, is also illustrated in both figures. 
For each pair of $M$ and $\beta$, we choose a random instance from the code ensemble and generate a random 
burst with length $b=\beta M$. 
The decoding failure probability, $P_\B^{\RBC}$, is averaged over 
all trials until 1000 decoding failures occur. We also plot the error floor estimation \eqref{eq:errorfloor} for each $M$.
Moreover, Fig.~\ref{fig:RBC_Bound} compares the error floor estimation with simulations for a larger range of $M$ and a fixed $\beta=1.25$. For each value of $M$, we performed Monte-Carlo simulations and counted $1000$ decoding failures  to assess the average block erasure probability $P_\B^{\RBC}$.

These figures show that for $b<\beta_{\rm max}M$, the error floor is well estimated by \eqref{eq:errorfloor} even for small $M=100$. It implies that the size-2 stopping sets are the main cause of decoding error. 
We also observe that the decoding error increases very fast for $b/M$ close to $\beta_{\rm max}$. As $M$ increases, the waterfall region becomes sharper around $\beta_{\rm max}$.
\end{example}

\begin{figure}[htb!]
\centering
\myfigure{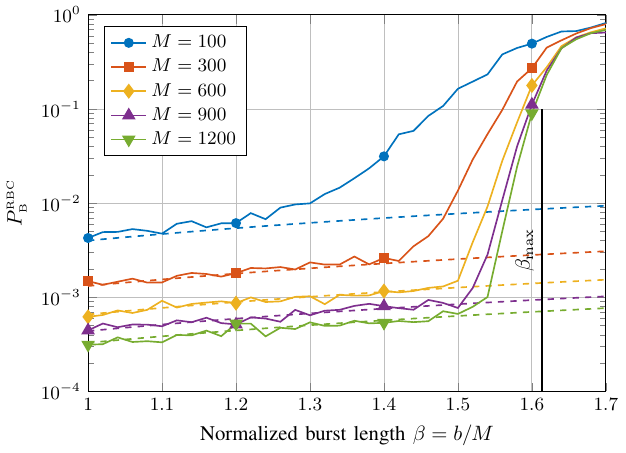}
\caption{Monte Carlo simulations over the RBC with the $\mathcal{C}_\mathcal{R}(3,6,3,100,M)$ ensemble with burst length $b=\beta M$. 
Solid lines represent simulation results and dashed lines the error bound \eqref{eq:errorfloor}.}\label{fig:finite363} 
\end{figure}

\begin{figure}[tbh!]
\myfigure{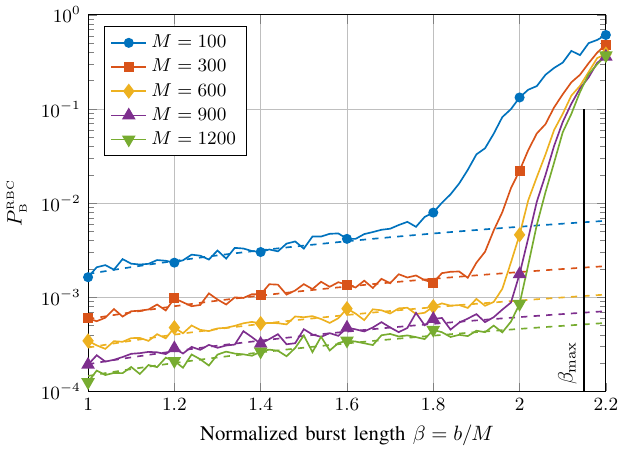}
\caption{Monte Carlo simulations over the RBC with the $\mathcal{C}_\mathcal{R}(3,6,4,100,M)$ ensemble with burst length $b=\beta M$. 
Solid lines represent simulation results and dashed lines the error bound \eqref{eq:errorfloor}.}\label{fig:finite364}
\end{figure}

\begin{figure}[tbh!]
\centering
\myfigure{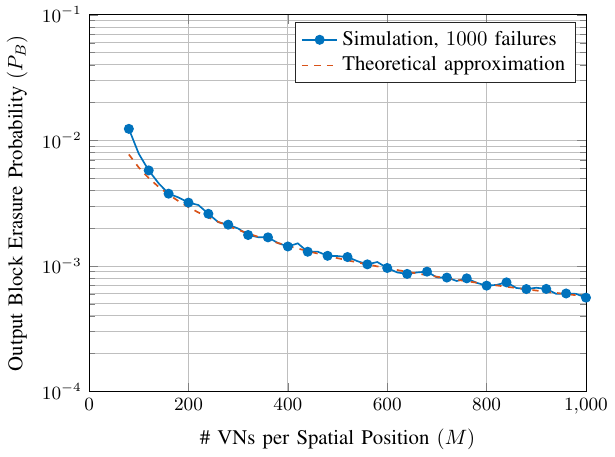}
\caption{Monte Carlo simulations over the RBC with the $\mathcal{C}_\mathcal{R}(3,6,3,100,M)$ ensemble with burst length $b=1.25M$, along with the approximation~\eqref{eq:errorfloor}.}\label{fig:RBC_Bound}
\end{figure}

\subsection{Single Burst Erasure in a BEC (SPBC+BEC)}
\label{sec:floorBEC}

In a more general scenario, we can consider a communication channel degraded by both memoryless noise and burst erasures.
We showed in~\cite{ARS_ISTC16} that SC-LDPC codes are asymptotically immune to a certain maximum burst length, which heavily depends on the distortion amount of the additional memoryless noise. Moreover, we observed that the behavior of these codes is very close on the BEC and the binary-input additive white Gaussian noise channel (BiAWGNC).
Here, for simplicity, we characterize the error floor when the channel is a BEC.
However, it is to be noted that the error floor on both these channels may not be the same.

\begin{sloppypar}
Consider a BEC with (random) erasure probability $\eps$ that, in addition, also introduces a single burst of erasures.
We assume that the burst erases a complete spatial position, i.e., a scenario similar to the SPBC. 
For instance, this setup can model a noisy distributed storage system with a single node failure and a noisy channel between storage nodes and users. 
As before, we focus on the number of erased size-2 stopping sets, $\mathbb{N}_2(\eps)$, because on the one hand,
\begin{equation*}
P_\B = \pr{\text{at least one stopping set erased}}\geq\pr{\mathbb{N}_2(\eps)\geq 1} ,
\end{equation*}
and on the other hand, $\pr{\mathbb{N}_2(\eps)\geq 1}$ converges to the expectation $\mathbb{E}[\mathbb{N}_2(\eps)]$ for large $ML$. Consider the $C_{\mathcal{R}}(d_v,d_c,w,L,M)$ ensemble in which
VNs of SP $z_0\in\{w,\ldots, L-w\}$ are erased by the burst error. The rest of the VNs are randomly erased with probability $\eps$.
We compute $\mathbb{E}[\mathbb{N}_2(\eps)]$ by summation of three distinct contributions:
\begin{itemize}
\item[(i)] Only BEC: A VN $v_i$ of SP $z$ and a VN $v_j$ of SP $z+k$ form a stopping set contributing to the error rate if and only if  both VNs are erased. The events of forming a stopping set and being erased are independent and thus the probability of having such an erased stopping set amounts to $\eps^2 q_k$, where $q_k$ is given by~\eqref{eq:prob_rbc}. Summing over all pairs of VNs (taking into account boundary effects), yields
\begin{equation*}
E^{(i)}=\eps^2\sum_{k=0}^{w-1} (L-k)\lambda_k,
\end{equation*}
where $\lambda_k$ is the average number of size-2 stopping sets given by \eqref{eq:lambdas}.
\item[(ii)] SPBC and BEC at SP $z_0$: The VNs of SP $z_0$ are all erased. We already counted the contribution of a fraction of these VNs in $E^{(i)}$ by assuming that the BEC erased the bits in SP $z_0$ equally likely. In particular, a pair of VNs in SP $z_0$ erased by the BEC will not be considered here.
Using \eqref{eq:prob_rbc}, the probability that a pair of VNs (in SP $z_0$) is not both erased by BEC, and
forms a stopping set is $(1-\eps^2)q_0$. Summing over all $\binom{M}{2}$ distinct pairs gives
\begin{equation*}
E^{(ii)}=(1-\eps^2)\lambda_0.
\end{equation*}
\item[(iii)] BEC and SPBC between SPs: A VN in SP $z_0$, erased by the SPBC and not by the BEC, can also form a stopping set with VNs erased by the BEC in SP $z_0\pm k$, $1\leq k<w$. 
The probability of such a stopping set
is $\eps(1-\eps)q_k$, and summing over all possible pairs gives
\begin{equation*}
E^{(iii)}=2\eps(1-\eps)\sum_{k=1}^{w-1} \lambda_k.
\end{equation*}
\end{itemize}
Summing all disjoint contributions leads to
\begin{align*}
\mathbb{E}[\mathbb{N}_2(\eps)]&=E^{(i)}+E^{(ii)}+E^{(iii)}\nonumber\\
&=\lambda_0 + \eps(2-\eps) \sum_{k=1}^{w-1} \lambda_k + \eps^2 \sum_{k=0}^{w-1} (L-k-1)\lambda_k
\end{align*}
and thus,
\begin{equation}
\label{eq:bec_spbc_ss2}
P_\B \geq\pr{\mathbb{N}_2(\eps)\geq 1}\approx \mathbb{E}[\mathbb{N}_2(\eps)].
\end{equation}
In case there is no contribution from the additional BEC (i.e., $\eps=0$), $\mathbb{E}[\mathbb{N}_2(0)]=\mathbb{E}[\Nsp]$ and the channel reduces to the SPBC model ($E^{(ii)}$). 
From Sec.~\ref{sec:spbcAnalysis}, we know that $P_\B \approx \mathbb{E}[\mathbb{N}_2(0)]$ for  large enough $M$. 
The same approximation is also valid if $\eps\ll \eps_{\rm \scriptscriptstyle BP}$, the BP threshold of the $C_{\mathcal{R}}(d_v,d_c,w,L,M)$ ensemble~\cite{Kudekar-it11}. 
The reason is that size-$2$ stopping sets become the dominant stopping sets when $\eps\ll \eps_{\rm \scriptscriptstyle BP}$  
and $M$ is large (see also~\cite{Olmos-isit11}). Moreover, $\mathbb{N}_2(\eps)$ converges to a Poisson distribution for large $M$. 
\end{sloppypar}

Note that $E^{(i)}=O(\varepsilon^2L)$ and it can become the dominant term when $L$ is large.
To see all the contributions, we compute the decoding failure of $\mathcal{C}_\mathcal{R}(3,6,3,L=10,M)$ ensemble over a combined SPBC and BEC. Two different erasure probabilities are considered: $\epsilon=0.1$ and $\epsilon=0.2$.
Fig.~\ref{fig:spbc_bec} illustrates the results of both simulations.
We plot the empirical block erasure probabilities as well as its approximation from~\eqref{eq:bec_spbc_ss2} in terms of $M$.
The simulation results are averaged over $1000$ decoder failures.
The results suggest that the approximation becomes tight beyond $M \gtrapprox 150$. 

\begin{figure}[t!]
\centering
\myfigure{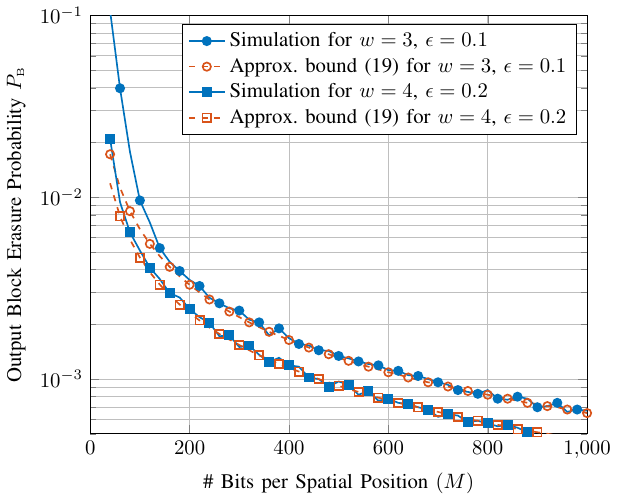}
\caption{Monte Carlo simulations over the SPBC with random erasures at rate $\epsilon$ with the $\mathcal{C}_\mathcal{R}(3,6,w,10,M)$ ensemble, along with the approximate bound~\eqref{eq:bec_spbc_ss2}.}\label{fig:spbc_bec}
\end{figure}

\begin{figure}[t!]
\centering
\myfigure{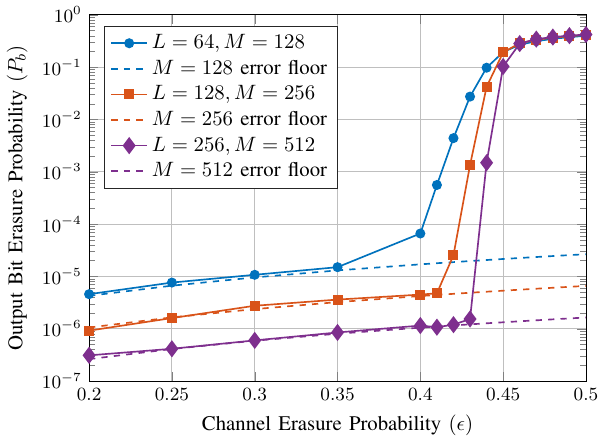}
\caption{Expected error floors for the $\mathcal{C}_\mathcal{R}(3,6,3,L,M)$ ensemble on the BEC.}
\label{fig:bec_floors}
\end{figure}

If we neglect the contributions of the SPBC (considering only $E^{(i)}$), 
we find an approximation for the error floor of the ensemble over a BEC. For $\eps\ll \eps_{\rm \scriptscriptstyle BP}$,
\begin{equation*}
P_B\approx \mathbb{E}[\mathbb{N}_2(\eps)] = E^{(i)}=\eps^2\sum_{k=0}^{w-1} (L-k)\lambda_k,
\end{equation*}
and the bit erasure probability of decoding is
\begin{equation}
\label{eq:errorfloor_bec}
P_b\approx \frac{2\mathbb{E}[\mathbb{N}_2(\eps)]}{LM}.
\end{equation}
Similar approximations are shown in~\cite{Olmos-isit11} for the error floor of particular 
protograph-based SC-LDPC ensembles over the BEC. 
Fig.~\ref{fig:bec_floors} shows the error floor of $\mathcal{C}_\mathcal{R}(3,6,3,L,M)$ ensemble in terms of $\eps$ with $M = 128,256,512$ and $L=M/2$. 
The dashed-lines are computed according to \eqref{eq:errorfloor_bec}. 
We see that the approximation becomes tight for small $\eps$ when, indeed, size-$2$ stopping sets dominate the performance.

\subsection{Block Erasure Channel (BLEC)}
\label{sec:blec}

Now we consider a more practical scenario where 
multiple bursts may simultaneously affect the transmission. 
The channel erases (all the VNs lying in) each SP independently with some probability $p$.
We call this model \emph{block erasure channel} and denote by BLEC($p$). 
This channel model is more challenging to analyze because there are $Lp$ SPs erased on average, with arbitrary locations. Two such erasures can be spaced closely together (within $w$ SPs) causing \emph{correlated} decoding failures or sufficiently spaced apart causing \emph{uncorrelated} decoding failures.
Our analysis consists of three steps:
\begin{itemize}
\item[(i)]
The first step is to partition the SPs into \textit{segments} of uncorrelated decoding failures. The key observation is that decoding failures from two consecutive erased SPs are independent if there are at least $w-1$ non-erased SPs between them. Thus, a segment is defined as 
the sequence of consecutive SPs such that the last 
$w-1$ SPs are non-erased and those are the only
$w-1$ consecutive non-erased SPs in the segment.  
Let $\tau_i$ denote the length of segment $i$. Since each SP is erased independently with probability $p$, $\tau_i$ is an i.i.d. random process with a bounded variance and average length given by the following lemma.
\begin{lemma}
\label{lem:blec_nw}
Consider the $\mathcal{C}_{\mathcal{R}}(d_v,d_c,w,L,M)$ SC-LDPC ensemble on the block erasure channel BLEC($p$). Assume that $L\gg w$. Define a segment as the sequence of consecutive SPs such that the last $w-1$ SPs are non-erased and those are the only $w-1$ consecutive non-erased SPs in the segment.  The expected number of spatial positions in a segment is given by
\begin{align}
\label{eq:blec_nw}
\mathbb{E}[\tau_i] & = \frac{1}{p} \left( \frac{1}{(1-p)^{w-1}} - 1 \right) \nonumber \\
 & \approx (w-1) \left( 1 + \frac{1}{2} (w-1)p + \frac{5}{6} (w-1)^2 p^2 \right) .
\end{align}
\end{lemma}
The proof is outlined in Appendix~\ref{sec:appblec}. 
\item[(ii)]
There are $N$ segments if we have\footnote{In the boundary, there are $w-1$ extra SPs with only check nodes. We can assume that those SPs have known VNs (non-erased) with trivial zero values.} $\sum_{i=1}^N\tau_i=L+w-1$. If $P_e^{(i)}$ denotes the decoding failure probability of segment $i$, the average probability of decoding failure is
\begin{equation*}
P_\B^\BLEC=1-\mathbb{E}\left[\prod_{i=1}^N \left(1-P_e^{(i)}\right)\right],
\end{equation*}
where the expectation is taken over different channel realizations and the code ensemble.
To simplify the above expression, we apply the tower rule of expectation to get
\begin{align}
P_\B^\BLEC&=1-\mathbb{E}_N\left[\mathbb{E}\left[
\prod_{i=1}^N \left(1-P_e^{(i)}\right) \bigg\vert N\right]\right]\nonumber\\
&=1-\mathbb{E}_N\left[\left(1-\mathbb{E}\left[P_e^{(1)}\right]\right)^N\right]\nonumber\\
&\overset{(i)}{\lessapprox} 1-\left(1-\mathbb{E}\left[P_e^{(1)}\right]\right)^{\mathbb{E}\left[N\right]}\nonumber\\
&\overset{(ii)}{\lessapprox} \mathbb{E}\left[N\right]\mathbb{E}\left[P_e^{(1)}\right]\nonumber\\
&=\frac{(L+w-1)}{\mathbb{E}[\tau_i]}\mathbb{E}\left[P_e^{(1)}\right],
\label{eq:pe_blec}
\end{align}
where $(i)$ follows from Jensen's inequality and $(ii)$ is because 
$(1-x)^k\geq 1-kx$ for $x\geq 0$ and the last identity is because $\mathbb{E}[N]=(L+w-1)/\mathbb{E}[\tau_i]$. The inequality $(i)$ becomes a good approximation
as $N$ concentrates to $\mathbb{E}[N]$ for large $L$, and the inequality $(ii)$
is a good approximation if $\mathbb{E}\left[N\right]\mathbb{E}\left[P_e^{(1)}\right]\ll 1$.
 
\item[(iii)]
Finally, we need to estimate $P_{\B, \rm seg}=\mathbb{E}\left[P_e^{(1)}\right]$,
the block error probability of a segment, which can be written as
\begin{equation*}
P_{\B,\rm seg} = \sum_{k=1}^{\infty} Q_k,
\end{equation*}
where $Q_k$ is the average probability of segment block error
when $k$ SPs (in arbitrary 
locations) have been erased. In fact, each $Q_k$ is the expectation of the
existence of a stopping set (SS) over all possible combinations of $k$ erased SPs in a segment.
While the probability of each of those combinations depends on $w$ and the block erasure probability $p$ of the channel,
the occurrence of an SS depends on the code ensemble and in particular~$M$.
We can simply estimate the dominant term of each $Q_k$ separately following the methodology introduced in the previous sections.

For ease of exposition, we exemplarily consider the specific ensemble $\mathcal{C}_\mathcal{R}(3,6,w=4,L,M)$. The results can be easily modified for other ensembles.
The DE analysis of a single-burst channel, shown in Fig.~\ref{fig:beta}, yields  
$\beta(0)=\beta_{\rm max}=2.14$ implying that BP decoding is not successful if three or more consecutive SPs are erased in the ensemble and thus, $Q_k>0$ for $k\geq 3$, even asymptotically.
For $k<3$, $Q_k$ tends to zero in the limit of $M$. 
In this case, we estimate $Q_k$ using our single-burst models.  The case of $k=1$
is basically the SPBC scenario when a single SP is erased times 
the probability of erasing only a SP in a segment. Therefore,
\begin{align*}
Q_1 &= P_{\B}^{\SPBC}(1-p)^{w-1} G(1)\\
    &= P_{\B}^{\SPBC}(1-p)^{w-1} (1 - (1-p)^{w-1}),
\end{align*}
where $G(x)$ is defined in Appendix~\ref{sec:appblec}. For $k=2$,
\begin{align}
Q_2 &= \sum_{z=0}^{w-2}\pr{\text{$z$ SPs between 2 erased SPs}}\nonumber\\
&\hspace{0.5cm}\times\pr{\text{having an SS}|\text{$z$ SPs between 2 erased SPs}}\nonumber\\
&\overset{(i)}{\gtrapprox} \pr{z=0}
\pr{\text{having an SS}|z=0}\nonumber\\
&= p(1-p)^{w-1} G(1)
\pr{\text{having an SS}|z=0}\label{eq:2spbc}\\
&\overset{(ii)}{\gtrapprox} p(1-p)^{w-1} G(1)
\left(2\binom{M}{2}q_0 + M^2q_1\right)\label{eq:2spbc_bound},
\end{align}
where $q_0$ and $q_1$ are obtained using Lemma~\ref{lem:prob_rbc} and $G(1)=1 - (1-p)^{w-1}$ as defined in Appendix~\ref{sec:appblec}. 
and the inequality $(i)$ is obtained by only considering the
dominant term of the sum, when 2 consecutive SPs are erased ($z=0$). We define 
$P^{\rm 2PBC}_{\rm B}=\pr{\text{having an SS}|z=0} = 2\binom{M}{2}q_0+M^2q_1$.
Let us justify $(ii)$ using the bound for the RBC. Consider a burst of length $b=2M$.
We bound the probability of error in \eqref{eq:errorfloor} by averaging over a random, uniformly distributed starting bit $S$.
Here, $S=1$. Hence, we have $(ii)$ by considering only the term for $S=1$ in \eqref{eq:errorfloor} (without the averaging factor $\frac{1}{M}$). 
  
Note that one can also estimate all remaining $w-1$ terms of $Q_2$ following the same methodology. We neglect these terms as \eqref{eq:2spbc} is simple and much larger than the other terms.  
For $k\geq 3$, BP decoding fails if there are at least 3 consecutive erased SPs out of $k$ erased SPs in a segment. 
These combinations become the dominant term of $Q_k$ for large values of $M$. Let $Q_k^{\rm cons}$ denote the probability of observing $k$ erased SPs with at least 3 consecutive erased SPs (in this case, there is certainly a stopping set and decoding is impossible). Then
\begin{align*}
 \sum_{k=3}^{\infty} Q_k &\gtrapprox \sum_{k=3}^{\infty}Q_k^{\rm cons}\\
 &=p^2\frac{1-(1-p)^{w-1}}{(1-p)^{w-1}+p},
\end{align*}
where the last equality is the result of Lemma~\ref{lem:q3} in Appendix~\ref{sec:appblec}.
As a result,
\begin{equation*}
P_{\B,\rm seg} \approx \underset{\rm finite-length\ error}{\underbrace{Q_1 + Q_2 }} + \underset{\rm DE\ error}{\underbrace{\sum_{k=3}^{\infty}Q_k^{\rm cons}}}.
\end{equation*}
Let us now summarize the above analysis. 
The error events are divided into two subsets: the first subset, named here ``finite-length error'', includes the error events with vanishing probability in the limit of $M$. The quantities $Q_1$ and $Q_2$ are its dominant terms. The second subset, named here ``DE error'', are the error events that cause the BP decoder to always fail for any value of $M$.
Combining steps (i), (ii), and (iii), we have,
\begin{align}
\label{eq:p_blec}
 P_{\B}^{\BLEC}\!\!&\approx (L+w-1)p(1-p)^{w-1} \bigg( (1-p)^{w-1} \big[ P_{\B}^{\SPBC} +   \nonumber \\ 
 &  p P^{\rm 2PBC}_{\rm B}  \big]
+\frac{p^2}{(1-p)^{w-1}+p} \bigg) ,
\end{align}
where $P_{\B}^{\SPBC}$ and $P^{\rm 2PBC}_{\rm B}$ are given in \eqref{eq:lower_bound} and 
\eqref{eq:2spbc_bound}, respectively.
\end{itemize}

\begin{figure}[tbh!]
\myfigure{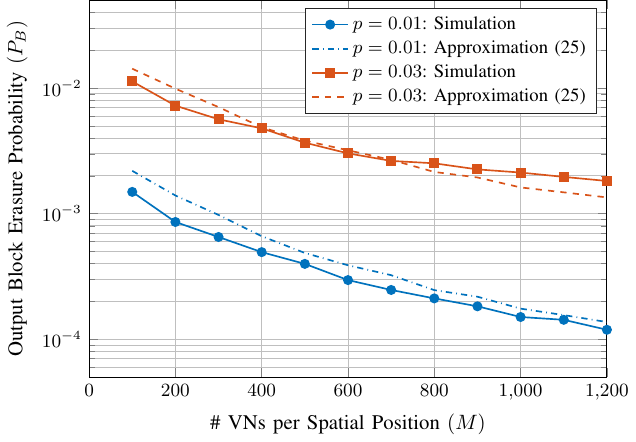}
\caption{Monte Carlo simulations over the BLEC with the $\mathcal{C}_{\mathcal{R}}(3,6,4,30,M)$ ensemble, along with their respective approximations according to~(\ref{eq:p_blec}).}\label{fig:BLEC_w_4}
\end{figure}

\begin{figure}[tbh!]
\myfigure{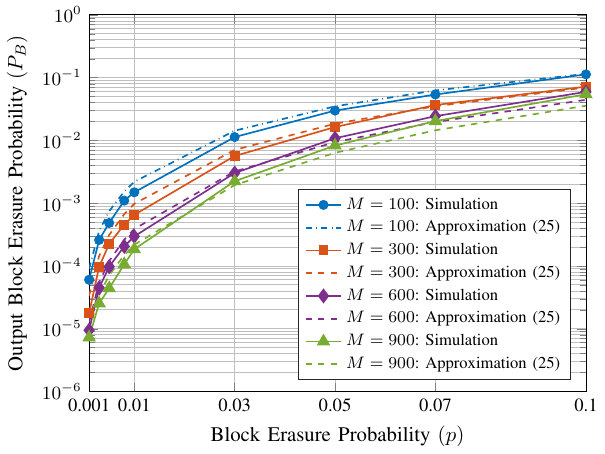}
\caption{Monte Carlo simulations over the BLEC with the $\mathcal{C}_{\mathcal{R}}(3,6,4,30,M)$ ensemble, along with their respective approximations according to~(\ref{eq:p_blec}).}\label{fig:BLEC_w_4_vary_p}
\end{figure}

We performed Monte Carlo simulations over the BLEC to verify the tightness of the above estimation for the $\mathcal{C}_{\mathcal{R}}(3,6,4,30,M)$ ensemble.
The simulation results are obtained by counting $100$ decoding failures.
Figure~\ref{fig:BLEC_w_4} illustrates the goodness of  
approximation \eqref{eq:p_blec} in terms of $M$ for two different values of $p=0.01$ and $p=0.03$.
We observe that the approximation well match the simulation results.
Note that the approximation requires the estimation of $P^{\rm 2PBC}_{\rm B}$ for a burst of length $2M$. For this purpose, we used computer simulations to numerically compute this value. Another option would be to use the
error-floor expression \eqref{eq:2spbc_bound}, which would however result in a worse approximation for small values of $M\approx 100$.

We also plot the average block error probability in terms of the channel parameter $p$ in Fig.~\ref{fig:BLEC_w_4_vary_p}. We see again that
the approximations track the actual performance closely for $p$ even as large as $0.1$ and for $M$ as small as $100$.

The above analysis indicates that
the error analysis of a more complex erasure channel model, such as the
BLEC$(p)$, is feasible by breaking down the channel model
into combinations of single-burst channel models.

\section{Conclusions}\label{sec:conclusions}

In this paper, we have analyzed random regular SC-LDPC ensembles on burst-erasure channels. 
Using density evolution, we have shown the asymptotically correctable maximum burst erasure length. 
However, for finite-length codes, some burst erasures shorter than this length may not be recoverable. 
The reason is that due to finite-length effects of certain graph structures, known from conventional LDPC decoding, the decoding could stop prematurely. 
Based on these structures, we have derived tight lower bounds or approximations on the block erasure probability of SC-LDPC code ensembles when the burst has one of the following characteristics:
\begin{itemize}
\item it erases a single random spatial position completely, or
\item it erases a random spatial position completely in addition to a memoryless background noise (e.g. BEC), or
\item it has a fixed length but starts at a random bit and erases a sequence of spatial positions, or
\item it erases each spatial position independently with some probability $p$. 
\end{itemize}
Moreover, we have shown that expurgating the codes can considerably improve the block erasure rates and guarantee virtually error-free performance (e.g. $P_{\B}^{\SPBC} \approx 10^{-15}$) for very short block lengths (e.g. $M \approx 120$).

Directions for future research include the construction of new, possibly more structured burst-resilient ensembles based on this analysis and a more detailed analysis of the block erasure channel.

\section*{Acknowledgment}

We would like to thank the reviewers and the associate editor for their valuable comments. Specifically, we extend our sincere thanks to one of the reviewers for encouraging us to consider an analysis of the block erasure channel.

\appendices
\allowdisplaybreaks

\section{Proof of Theorem~\ref{thm:thmlb}}
\label{sec:appProof1}

Before presenting the proof of Theorem~\ref{thm:thmlb}, we give the following Lemma, which we will need in our proof besides Lemma~\ref{lem:PU_indic}.
\begin{lemma}\label{lem:PU_joint}
Consider one SP of a code sampled from the $\mathcal{C}_{\mathcal{R}}(d_v,d_c,w,L,M)$ SC-LDPC ensemble. Recall the indicator function $U_{ij}=1$ if VNs $v_i$ and $v_j$ in that SP form a size-2 stopping set. Assuming $d_c > 2$ and $wM \geq 2(d_v+1)d_c$, we have the following bound on the joint expectation
\begin{align*}
\mathbb{E}[U_{ij}U_{kl}]\leq \frac{2\mathbb{E}[U_{ij}]}{\binom{wM\frac{d_v}{d_c}-2d_v}{d_v}}\,.
\end{align*}
\end{lemma}
\begin{proof}
We have to distinguish two cases. 
In the first case,  variable nodes $v_i$ and $v_j$, $j\neq i$ form a stopping set of size two and variable nodes $v_k$ and $v_l$, $l\neq k\neq i$ form another stopping set, where w.l.o.g., we assume $k=j$, i.e., the two stopping sets share a variable node. In this case, we have
\begin{align*}
\mathbb{E}\left[U_{ij}U_{jl}\right] &= \pr{U_{ij}=1,U_{jl}=1}\\
&= \pr{U_{jl}=1|U_{ij}=1}\pr{U_{ij}=1} \\
&= \mathbb{E}[U_{ij}]\cdot \pr{U_{jl}=1|U_{ij}=1}
\end{align*}
where $\pr{U_{jl}=1|U_{ij}=1}$ can be computed using similar arguments as in the proof of Lemma~\ref{lem:PU_indic}, except that some sockets have already been occupied by the edges emerging from the VNs $v_i$ and $v_j$ and for forming an SS, node $v_l$ can choose among $d_v(d_c-2)$ sockets. We have
\begin{equation*}
\pr{U_{jl}=1|U_{ij}=1} =\\
 \frac{(1-2/d_c)^{d_v}}{\sum_{\ell=0}^{d_v}\binom{d_v}{\ell}\binom{wM\frac{d_v}{d_c}-d_v}{d_v-\ell}(1-2/d_c)^\ell}.
\end{equation*}
We upper bound $\pr{U_{jl}=1|U_{ij}=1}$ by considering only the dominant term ($\ell=0$) in the denominator
\begin{align*}
\pr{U_{jl}=1|U_{ij}=1} &\leq \frac{(1-2/d_c)^{d_v}}{\binom{wM\frac{d_v}{d_c}-d_v}{d_v}}
\leq \frac{1}{\binom{wM\frac{d_v}{d_c}-d_v}{d_v}}.
\end{align*}
Hence, for $i\in[M]$, $j\in[M]\backslash\{i\}$ and $l\in[M]\backslash\{i,j\}$ we have
\[
\mathbb{E}\left[U_{ij}U_{jl}\right] \leq \frac{\mathbb{E}[U_{ij}]}{\binom{wM\frac{d_v}{d_c}-d_v}{d_v}}.
\]
Now, we consider the second, more general case, where $i\in[M]$, $j\in[M]\backslash\{i\}$ and $k\in[M]\backslash\{i,j\}$ and $l\in[M]\backslash\{i,j,k\}$. If $v_i$ and $v_j$ form a stopping set of size two, we have $|\mathcal{N}(v_i)\cap\mathcal{N}(v_j)| = d_v$. Thus, we have
\begin{align*}
 & \mathbb{E}\left[U_{ij}U_{kl}\right] \\
 & = \pr{U_{ij}=1} \pr{U_{kl}=1|U_{ij}=1}\\
&\! \stackrel{(a)}= \pr{U_{ij}=1} \sum_{b=0}^{d_v-1} \pr{U_{kl}=1, |\mathcal{N}(v_i)\cap\mathcal{N}(v_k)|=b \big\vert U_{ij}=1} \\
& \triangleq \pr{U_{ij}=1}\sum_{b=0}^{d_v-1} P_{\text{corr},b} ,
\end{align*}
where in $(a)$ we marginalize over the number $b$ of check nodes shared between both size-2 stopping sets $\{v_i,v_j\}$ and $\{v_k,v_l\}$, $b\leq d_v$. Let $T_{\text{ss},b}$ denote the number of 
favorable sub-graphs from these VNs provided that $b$ check nodes between both stopping sets are shared and let $\tilde{T}$ denote the number of all sub-graphs from these VNs provided that
$U_{ij}=1$ and $|\mathcal{N}(v_i)\cap\mathcal{N}(v_l)\cap \mathcal{N}(v_k)|\leq d_v$. Thus, 
\[
P_{\text{corr},b}\triangleq\! \pr{U_{kl}\!=\!1, |\mathcal{N}(v_i)\cap\mathcal{N}(v_k)|=b \big\vert U_{ij}=1}=T_{\text{ss},b}/\tilde{T}.
\]
First, we compute $T_{\text{ss},b}=T_{\{i,j\},k,b}T_{\{i,j,k\},l,b}$ as follows:
$T_{\{i,j\},k,b}$ is the number of sub-graphs in which $v_k$ is connected to randomly chosen $b$ CNs of $\mathcal{N}(v_i)$, i.e., CNs  with already $2$ connections from $v_i$ and $v_j$.
$T_{\{i,j,k\},l,b}$ is the number of ways to connect $v_l$ to any of these sub-graphs such that $U_{kl}=1$. We have
\begin{align*}
T_{\{i,j\},k,b} &= d_v!\binom{d_v}{b}\binom{wM\frac{d_v}{d_c}-d_v}{d_v-b}(d_c-2)^{b}(d_c)^{d_v-b},
\end{align*}
where $d_v!$ is due to the permutation of labeled edges, and the rest is counting the different ways of choosing labeled sockets of CNs.
Similarly, we get
\begin{align*}
T_{\{i,j,k\},l,b} &= d_v!(d_c-3)^b (d_c-1)^{d_v-b}.
\end{align*}
This yields
\begin{align*}
& T_{\text{ss},b} \\
&= T_{\{i,j\},k,b}T_{\{i,j,k\},l,b}  \\
&= (d_v!)^2\binom{d_v}{b}\binom{wM\frac{d_v}{d_c}-d_v}{d_v-b}\left(d_c(d_c-1)\right)^{d_v-b} \\
& \hspace{5.5cm} \times\left((d_c-2)(d_c-3)\right)^b\\
&\leq (d_v!)^2\binom{d_v}{b}\binom{wM\frac{d_v}{d_c}-d_v}{d_v-b}
\left(d_c(d_c-1)\right)^{d_v}\left(\frac{d_c-2}{d_c}\right)^{2b}.
\end{align*}
The total number of possibilities $\tilde{T}$ can be lower bounded by its subset which is $|\mathcal{N}(v_i)\cup\mathcal{N}(v_k)\cup \mathcal{N}(v_l)|=3d_v$ (no common CNs).
\begin{align*}
\tilde{T} &\geq (d_v!)^2\binom{wM\frac{d_v}{d_c}-d_v}{d_v}\binom{wM\frac{d_v}{d_c}-2d_v}{d_v}(d_c)^{2d_v}.
\end{align*}
Hence $P_{\text{corr},b}$ can be bounded as
\begin{align*}
P_{\text{corr},b}  &=  \frac{T_{\text{ss},b}}{T_b} \\
  & \leq  \frac{\binom{d_v}{b}\binom{wM\frac{d_v}{d_c}-d_v}{d_v-b}}{\binom{wM\frac{d_v}{d_c}-d_v}{d_v}\binom{wM\frac{d_v}{d_c}-2d_v}{d_v}} \left(1-\frac{1}{d_c}\right)^{d_v} \left(1-\frac{2}{d_c}\right)^{2b}\\
  &\stackrel{(a)}\leq
  \frac{\binom{d_v}{b}d_c^b}{(wM-2d_c)^b\binom{wM\frac{d_v}{d_c}-2d_v}{d_v}} \left(1-\frac{1}{d_c}\right)^{d_v} \left(1-\frac{2}{d_c}\right)^{2b},
\end{align*}
where $(a)$ is as $\binom{wM\frac{d_v}{d_c}-d_v}{d_v-b}\leq \binom{wM\frac{d_v}{d_c}-d_v}{d_v}\left(\frac{d_v}{wM\frac{d_v}{d_c}-2d_v}\right)^b$. Finally, we get
\begin{align*}  
\sum_{b=0}^{d_v-1} P_{\text{corr},b} & \leq  \frac{(1-1/d_c)^{d_v}}{\binom{wM\frac{d_v}{d_c}-2d_v}{d_v}}\sum_{b=0}^{d_v} \binom{d_v}{b} \left(\frac{d_c-2}{wM - 2 d_c}\right)^{b} \\
  & =\frac{(1-1/d_c)^{d_v}}{\binom{wM\frac{d_v}{d_c}-2d_v}{d_v}} \left( 1 + \dfrac{d_c-2}{wM - 2 d_c} \right)^{d_v}\\
  &\overset{(a)}{\leq}  \frac{2}{\binom{wM\frac{d_v}{d_c}-2d_v}{d_v}}
\end{align*}
where $(a)$ is because $(1+\frac{d_c-2}{wM-2d_c})^{d_v}\leq (1+\frac{1}{2d_v})^{d_v}<2$ for $wM\geq 2(d_v+1)d_c$. Taking the larger upper-bound of both cases proves the claim.
\end{proof}

\subsection*{Proof of Theorem~\ref{thm:thmlb}:}
Consider one SP of a code sampled from the $\mathcal{C}_{\mathcal{R}}(d_v,d_c,w,L,M)$ ensemble.
Recall that $\mathbb{N}_2^{\rm\scriptscriptstyle SP}=\sum_{1\leq i<j\leq M} U_{ij}$ and $\mathbb{E}[\mathbb{N}_2^{\rm\scriptscriptstyle SP}]=\binom{M}{2}\pr{U_{ij}=1} = \binom{M}{2}\mathbb{E}[U_{ij}]$. We then have 
\begin{align}
P_\B^\SPBC & =  \pr{\text{At least one stopping set in a SP}}\nonumber \\
			& \geq  \pr{\Nsp \geq 1} \nonumber \\
			& \overset{(a)}{\geq}  \frac{\mathbb{E}[\Nsp]^2}{\mathbb{E}\left[\left(\Nsp\right)^2\right]}\label{eq:thm1prooft1}		
\end{align}
where $(a)$ is the application of the second moment method~\cite[Lemma C.8]{Richardson-MCT08}. Furthermore,
\begin{align*}
\mathbb{E}\left[\left(\Nsp\right)^2\right] &=
\mathbb{E}\left[ \left(\sum_{1\leq i<j\leq M} U_{ij}\right)^2\right]\\
	&= \sum_{1\leq i<j\leq M} \mathbb{E}[U_{ij}^2]
	+ \sum_{\substack{(i,j)\neq(k,l)\\ i<j,k<l}} \mathbb{E}[U_{ij}U_{kl}]\\
	&\leq \binom{M}{2}\mathbb{E}[U_{ij}] + \frac{2\binom{M}{2} \left( \binom{M}{2}-1 \right)}{\binom{wM\frac{d_v}{d_c}-2d_v}{d_v}}\mathbb{E}[U_{ij}] ,
\end{align*}
where in the last step, $\sum_{1\leq i<j\leq M} \mathbb{E}[U_{ij}^2]=\binom{M}{2}\mathbb{E}[U_{ij}]$ as
$U_{ij}\in\{0,1\}$ and the second term is the application of Lemma~\ref{lem:PU_joint} over the remaining $\binom{M}{2} \left( \binom{M}{2}-1 \right)$
combinations. 
Finally, we can write~\eqref{eq:thm1prooft1} as
\begin{align*}
P_\B^\SPBC &\geq \frac{\binom{M}{2}^2\mathbb{E}[U_{ij}]^2}{\binom{M}{2}\mathbb{E}[U_{ij}] + \frac{2\binom{M}{2} \left( \binom{M}{2}-1 \right)}{\binom{wM\frac{d_v}{d_c}-2d_v}{d_v}}\mathbb{E}[U_{ij}]} \\
& \geq \frac{\binom{M}{2}\mathbb{E}[U_{ij}]}{1 + \frac{2\binom{M}{2}}{\binom{wM\frac{d_v}{d_c}-2d_v}{d_v}}} \\
& \stackrel{(a)}{\geq} \mathbb{E}[\Nsp]\left(1 - \frac{2\binom{M}{2}}{\binom{wM\frac{d_v}{d_c}-2d_v}{d_v}}\right) \\
& \geq \mathbb{E}[\Nsp]\left(1 - \frac{M^2}{\left(\frac{wM}{d_c}-3\right)^{d_v}}\right) 
\end{align*}
where  $(a)$ is due to the fact that $\frac{1}{1+\tau}\geq 1-\tau$ for $\tau > -1$.\hfill\qed

\section{The Poisson Ensemble}
\label{sec:poisson}

\begin{lemma}\label{lem:PU_poisson}
Consider a code sampled uniformly from the Poisson ensemble $\mathcal{C}_{\mathcal{P}}(d_v,d_c,w,L,M)$. 
The probability that two variable nodes from a spatial position form a stopping set amounts to
\begin{equation}
\label{eq:p2poisson}
P_{\mathcal{P}} = \binom{wM\frac{d_v}{d_c}}{d_v}^{-1}.
\end{equation}
\end{lemma}
\begin{proof}
Let $v_i$ and $v_j$ be two VNs randomly chosen from an SP $z$ of the $\mathcal{C}_{\mathcal{P}}(d_v,d_v,w,L,M)$ ensemble, i.e., $i,j \in \{(z-1)M+1,\ldots,zM\}$ for $z \in [L]$.
The computation of $P_{\mathcal{P}} \triangleq \mathbb{E}[U_{ij}]$  for this ensemble is much simpler as we do not need to distinguish sockets. 
Assume that the edges of $v_j$ are assigned to CNs sequentially after edges of $v_i$ have been assigned to CNs. 
The first edge can connect to any of the $(wM\frac{d_v}{d_c})$ CNs from SPs $z,z+1,\ldots,z+w-1$. 
As we avoid parallel edges in the construction, the second edge has one CN less to choose from, the third edge has two CNs less to choose from and so on. 
But, there is exactly one way in which the edges can connect exactly to the same CNs as $v_i$, with $d_v!$ possible permutations of the edge arrangements. 
Hence the probability of $v_j$ forming a stopping set with $v_i$ is
\begin{align*}
P_{\mathcal{P}} & \triangleq \pr{U_{ij}=1} \\
 & = \frac{d_v!}{\left(wM\frac{d_v}{d_c}\right)\left(wM\frac{d_v}{d_c}-1\right)\cdots \left(wM\frac{d_v}{d_c}-d_v+1\right) }.
\end{align*}
which leads to the statement after simplification.
\end{proof}

\begin{theorem}\label{thm:SC_Poisson_Comp}
Consider the ensembles $\mathcal{C}_\mathcal{R}(d_v,d_c,w,L,M)$ and $\mathcal{C}_{\mathcal{P}}(d_v,d_c,w,L,M)$. 
For two VNs from the same spatial position, the probability that they form a stopping set is larger for the Poisson ensemble, i.e., $P_{\mathcal{P}} \geq P_{\mathcal{R}}$.
\end{theorem}
\begin{proof}
We upper bound~\eqref{eq:lem1_pdef} as
\begin{align*}
P_{\mathcal{R}} &= \frac{\left(1-\frac{1}{d_c}\right)^{d_v}}{\sum\limits_{\ell=0}^{d_v} \binom{d_v}{\ell}
\binom{wM\frac{d_v}{d_c}-d_v}{d_v-\ell} \left(1-\frac{1}{d_c}\right)^\ell} \\
&\leq \frac{1}{\sum\limits_{\ell=0}^{d_v} \binom{d_v}{\ell}
\binom{wM\frac{d_v}{d_c}-d_v}{d_v-\ell}} \stackrel{(a)}{=} \frac{1}{\binom{wM\frac{d_v}{d_c}}{d_v}} = P_{\mathcal{P}},
\end{align*}
where $(a)$ is due to Vandermonde's identity.
\end{proof}

\begin{lemma}\label{lem:PUpoisson_joint}
Consider one SP of a code sampled from the $\mathcal{C}_{\mathcal{P}}(d_v,d_c,w,L,M)$ SC-LDPC ensemble. Define the indicator function $U_{ij}=1$ if VNs $v_i$ and $v_j$ in that SP form a size-2 stopping set. If $wM \geq 2(d_v+1)d_c$, we have the following bound on the joint expectation
\begin{align*}
\mathbb{E}[U_{ij}U_{kl}] \leq \frac{2\mathbb{E}[U_{ij}]}{\binom{wM\frac{d_v}{d_c}-d_v}{d_v}}
\end{align*}
\end{lemma}
\begin{proof} We follow lines of the proof of Lemma~\ref{lem:PU_joint}. Again, we distinguish two cases.
In the first case,  variables $v_i$ and $v_j$, $j\neq i$, form a stopping set of size two and variable nodes $v_k$ and $v_l$ form another stopping stop, where w.l.o.g., we assume $k=j$, i.e., the two stopping sets share a variable node. In this case, we have
\begin{align*}
\mathbb{E}\left[U_{ij}U_{jl}\right] = \mathbb{E}[U_{ij}]\cdot \pr{U_{jl}=1|U_{ij}=1}
\end{align*}
where $\pr{U_{jl}=1|U_{ij}=1}$ can be computed using similar arguments as in the proof of Lemma~\ref{lem:PU_poisson} and we have
\[
\pr{U_{jl}=1|U_{ij}=1} = \frac1{\binom{wM\frac{d_v}{d_c}}{d_v}}.
\]
In the second, more general case, where $i\in[M]$, $j\in[M]\backslash\{i\}$ and $k\in[M]\backslash\{i,j\}$ and $l\in[M]\backslash\{i,j,k\}$, we have with $P_{\mathcal{P}} \triangleq \pr{U_{ij}=1}$,
\begin{align*}
 & \mathbb{E}\left[U_{ij}U_{kl}\right] \\
 &= P_{\mathcal{P}}\cdot \pr{U_{kl}=1|U_{ij}=1}\\
&\stackrel{(a)}= P_{\mathcal{P}}\sum_{b=0}^{d_v} \pr{U_{kl}=1, |\mathcal{N}(v_i)\cap\mathcal{N}(v_k)|=b \mid U_{ij}=1} \\
&\triangleq P_{\mathcal{P}}\sum_{b=0}^{d_v} P_{\mathcal{P},\text{corr},b} ,
\end{align*}
where in $(a)$ we marginalize over the number $b$ of check nodes shared by both size-2 stopping sets $\{v_i,v_j\}$ and $\{v_k,v_l\}$, $b\leq d_v$. 
Let $T_{\text{ss},b}$ denote the number of 
favorable sub-graphs from these VNs provided that $b$ check nodes between both stopping sets are shared and let $\tilde{T}$ denote the number of all sub-graphs from these VNs provided that
$U_{ij}=1$ and $|\mathcal{N}(v_i)\cap\mathcal{N}(v_l)\cap \mathcal{N}(v_k)|\leq d_v$. Thus, 
\[
P_{\mathcal{P},\text{corr},b} \triangleq \! \pr{U_{kl}\!=\!1, |\mathcal{N}(v_i)\cap\mathcal{N}(v_k)|=b \big\vert U_{ij}\!=\!1}\!=\!T_{\text{ss},b}/\tilde{T}.
\]
First, we compute $T_{\text{ss},b}=T_{\{i,j\},k,b}T_{\{i,j,k\},l,b}$ as follows:
$T_{\{i,j\},k,b}$ is the number of sub-graphs in which $v_k$ is connected to randomly chosen $b$ CNs of $\mathcal{N}(v_i)$, i.e., CNs  with already $2$ connections from $v_i$ and $v_j$.
$T_{\{i,j,k\},l,b}= d_v!$ is the number of ways to connect $v_l$ to $\mathcal{N}(v_k)$, i.e., $U_{kl}=1$. We have
\begin{align*}
T_{\{i,j\},k,b} &= d_v!\binom{d_v}{b}\binom{wM\frac{d_v}{d_c}-d_v}{d_v-b},
\end{align*}
where $d_v!$ is due to the permutation of labeled edges, and the rest of the expression counts the different ways of choosing CNs. This yields
\begin{align*}
T_{\text{ss},b} = (d_v!)^2\binom{d_v}{b}\binom{wM\frac{d_v}{d_c}-d_v}{d_v-b}.
\end{align*}
The total number of possibilities $\tilde{T}$ can be lower bounded by its subset which is $|\mathcal{N}(v_i)\cup\mathcal{N}(v_k)|=|\mathcal{N}(v_i)\cup \mathcal{N}(v_l)|=2d_v$ (no common CNs with $v_i$ and $v_j$), with
\begin{align*}
\tilde{T} &\geq \left(d_v!\binom{wM\frac{d_v}{d_c}-d_v}{d_v}
\right)^2.
\end{align*}
Using Vandermonde's identity, we have,
\begin{align*}
\sum_{b=0}^{d_v} P_{\mathcal{P}\text{corr},b}
&\leq \frac{1}{\binom{wM\frac{d_v}{d_c}-d_v}{d_v}
^2}\sum_{b=0}^{d_v} \binom{d_v}{b}\binom{wM\frac{d_v}{d_c}-d_v}{d_v-b}\\
&\leq \frac{\binom{wM\frac{d_v}{d_c}}{d_v}}{\binom{wM\frac{d_v}{d_c}-d_v}{d_v}^2}.
\end{align*}
We further have
\begin{align*}
\frac{\binom{wM\frac{d_v}{d_c}}{d_v}}{\binom{wM\frac{d_v}{d_c}-d_v}{d_v}}
&\leq \left(\frac{wM\frac{d_v}{d_c}-d_v}{wM\frac{d_v}{d_c}-2d_v}\right)^{d_v} \\
& \leq \left(1+\frac{1}{2d_v}\right)^{d_v} 
\leq \sqrt{e}<2.
\end{align*}
where the second inequality is because $wM\geq 2(d_v+1)d_c$ by assumption. Thus, 
\begin{align*}
\sum_{b=0}^{d_v} P_{\mathcal{P}\text{corr},b}
\leq \frac{2}{\binom{wM\frac{d_v}{d_c}-d_v}{d_v}}.
\end{align*}
The statement of the lemma follows by combining both cases.
\end{proof}

\begin{theorem}\label{thm:thmlb_poisson}
Consider a code sampled uniformly from the $\mathcal{C}_{\mathcal{P}}(d_v,d_c,w,L,M)$ ensemble with $wM \geq d_c^2$.
If a randomly chosen spatial position of this code is completely erased, the (average) probability of BP decoding failure is lower-bounded by 
\begin{align}
P_\B^\SPBC \geq \binom{M}{2}\left(1 - \frac{M^2}{(\frac{w}{d_c}M-2)^{d_v}}\right)P_{\mathcal{P}} , \label{eq:lower_bound_poisson}
\end{align}
where $P_{\mathcal{P}}$ is the probability that two variable nodes from  a spatial position of the code form a stopping set, given by~\eqref{eq:p2poisson}.
\end{theorem}
\begin{proof}
We follow a similar argument as in the proof of Theorem~\ref{thm:thmlb}. We get with Lemma~\ref{lem:PUpoisson_joint}
\begin{align*}
P_\B^\SPBC &\geq \frac{\binom{M}{2}\mathbb{E}[U_{ij}]}{1 + \frac{\binom{M}{2}}{\binom{wM\frac{d_v}{d_c}-2d_v}{d_v}}} \geq \mathbb{E}[\Nsp]\left(1 - \frac{M^2}{\left(\frac{wM}{d_c}-2\right)^{d_v}}\right)
\end{align*}
and by computing $\mathbb{E}[\Nsp]$ as $\binom{M}{2}P_{\mathcal{P}}$ using Lemma~\ref{lem:PU_poisson}.\hfill
\end{proof}

\section{}
\label{sec:appblec}
\subsection*{Proof of Lemma~\ref{lem:blec_nw}:}
We use the method of generating functions to compute this quantity. 
Specifically, we construct a function $G(x)$ where the exponent of $x$ in each term signifies the number of SPs involved in a segment with one erased SP. 
For convenience, we exclude the final sequence of $w-1$ unerased SPs in $G(x)$.
So finally we will have to multiply $G(x)$ by $(1-p)^{w-1} x^{w-1}$ to account for these.
Let $E$ denote an erased SP and $N$ denote a non-erased SP.
Then, starting from the first SP in a segment, we have the following possible outcomes with their respective probabilities and generating function terms:
\begin{center}
\begin{tabular}{rcc}
Pattern & Probability & Term in $G(x)$ \\
\hline
$E$ & $p$ & $px$ \\
$NE$ & $(1-p) p$ & $(1-p) p x^2$ \\
$NNE$ & $(1-p)^2 p$ & $(1-p)^2 p x^3$ \\
$\vdots$ & $\vdots$ & $\vdots$ \\
$\underset{w-2}{\underbrace{NN\ldots N}} E$ & $(1-p)^{w-2} p$ & $(1-p)^{w-2} p x^{w-1}$ \\
\end{tabular}
\end{center}
So we have
\begin{align*}
G(x) &= px + (1-p) p x^2 + \ldots + (1-p)^{w-2} p x^{w-1} \\
  &= px \cdot \frac{1 - (1-p)^{w-1} x^{w-1}}{1 - (1-p)x}
\end{align*}
Note that $G(x) < 1 \ \forall \ x \leq 1$ and $G(1) = 1 - (1-p)^{w-1}$.
Also for $k \geq 0$, $G(x)^k$ generates all the possible sequences with $k$ SPs erased and where the separation between two consecutive erased SPs is strictly less than $w-1$.
Hence every possible sequence ending with $w-1$ $N$s can be enumerated as
\[ F(x) = \sum_{k=0}^{\infty} G(x)^k (1-p)^{w-1} x^{w-1} = (1-p)^{w-1} x^{w-1} \frac{1}{1 - G(x)} . \]
Note that $F(1)=1$ which implies that we have counted all the possible sequences.
Also, implicitly we have set $L \rightarrow \infty$ for convenience.
In fact, we can rewrite $F(x)$ as
\[ F(x) = \sum_{n=w-1}^{\infty} p_n x^n , \]
where $p_n$ is the probability of all sequences of length $n$.
Then we can compute the required quantity as 
\[ \mathbb{E}[\tau_i] = \frac{d}{dx} F(x)\biggr\rvert_{x=1} = F'(1) = \frac{1}{p} \left( \frac{1}{(1-p)^{w-1}} - 1 \right).\qed \]

\begin{lemma}
\label{lem:q3}
Assume the transmission of a code from the $\mathcal{C}_{\mathcal{R}}(d_v,d_c,w,L=\infty,M)$ SC-LDPC ensemble over the block erasure channel BLEC($p$). Consider a segment of spatial positions, as defined in Section~\ref{sec:blec}. Let $Q_k^{\rm cons}$ denote the probability of incurring $k$ erased SPs with at least 3 consecutive erased SPs. Then
\begin{align*}
\sum_{k=3}^{\infty}Q_k^{\rm cons}=p^2\frac{1-(1-p)^{w-1}}{(1-p)^{w-1}+p}.
\end{align*}
\begin{proof}
We need to count all possible combinations of erased SPs in a segment
with at least $3$ consecutive erased SPs. We use again the method of generating functions with the same $G(x)$ as defined in the proof of Lemma~\ref{lem:blec_nw}. All possible combinations are depicted as follows:

\begin{align*}
\underset{k_1 \geq 0}{
\underbrace{
\left\{
\begin{array}{c}
E \\
NE \\
NNE \\
\vdots \\
N\ldots NE 
\end{array}
\right\}
\cdots
\left\{
\begin{array}{c}
E \\
NE \\
NNE \\
\vdots \\
N\ldots N E
\end{array}
\right\}
}}
\left\{
\begin{array}{c}
E \\
NE \\
NNE \\
\vdots \\
N\ldots N E
\end{array}
\right\}
EE \\
\underset{k_2 \geq 0}{
\underbrace{
\left\{
\begin{array}{c}
NE \\
NNE \\
\vdots \\
N\ldots N E
\end{array}
\right\}
\cdots
\left\{
\begin{array}{c}
NE \\
NNE \\
\vdots \\
N\ldots N E
\end{array}
\right\}
}}
\underset{w-1}{\underbrace{NN \ldots N}} .
\end{align*}
The middle terms guarantee to have 3 consecutive erased SPs. 
The $k_1$ and $k_2$ enumerate all possible erasure sequences 
that can happen before and after the consecutive erased SPs. 
Note that we must exclude the event of ``E'' (only single erased SP) in the $k_2$ sequences (or, alternatively, in the $k_1$ sequences) to generate all possible combinations once. 

As a result,
\begin{align*}
\sum_{k=3}^{\infty}Q_k^{\rm cons} &=
\left( \sum_{k_1=0}^{\infty} G(x)^{k_1} \right) G(x) (px)^2 \\
 & \qquad \times\left( \sum_{k_2=0}^{\infty} (G(x)-px)^{k_2} \right) (1-p)^{w-1} x^{w-1} \biggr\vert_{x=1}\\
&=\frac{1}{1-G(1)}G(1)p^2\frac{1}{1-G(1)+p}(1-p)^{w-1} \\
&=p^2\frac{1-(1-p)^{w-1}}{(1-p)^{w-1}+p}. \qedhere
\end{align*}

\end{proof}

\end{lemma}
\newpage

\begin{IEEEbiographynophoto}{Vahid Aref} is a member of technical staff in Nokia Bell Labs. 
He received B.Sc. and M.Sc. degrees in Electrical Engineering from Sharif University of Technology, Iran.
He received his PhD degree in computer and communication sciences from \'{E}cole Polytechnique F\'{e}d\'{e}rale de Lausanne (EPFL), Lausanne, Switzerland, in 2014. Later, he conducted post-doctoral research in the institute of telecommunications (IN\"{U}) at the University of Stuttgart for a year before joining Bell Labs in 2015. Since 2016, Dr. Aref also serves as guest lecturer at the University of Stuttgart.
\end{IEEEbiographynophoto}

\begin{IEEEbiographynophoto}{Narayanan Rengaswamy}
received the Bachelor of Technology (B.Tech.) degree in Electronics and Communication Engineering from Amrita University, Coimbatore, India, where he secured the third rank at the university level, across three campuses. Subsequently, in 2015, he obtained the Master of Science (M.S.) degree from Texas A\&M University, College Station, TX, USA where he worked on cyclic polar codes under the supervision of Prof. Henry D. Pfister. During the summer of 2015 he was a graduate research intern at Alcatel-Lucent Bell Labs, Stuttgart, Germany where he worked on spatially-coupled LDPC codes. Currently he is pursuing his Doctor of Philosophy (Ph.D.) studies at Duke University, Durham, NC, USA under the supervision of Prof. Henry D. Pfister and Prof. Robert Calderbank. His research interests are in classical and quantum coding theory, quantum computing and communications, information theory and compressed sensing.
\end{IEEEbiographynophoto}

\begin{IEEEbiographynophoto}{Laurent Schmalen}
received both his Dipl.-Ing. degree in electrical engineering and information technology and his Dr.-Ing. degree from the RWTH Aachen University of Technology. In 2011, he joined Bell Labs as a member of technical staff where he now heads the department of coding in optical communications within the IP and Optical Transport research lab. He is a two-time recipient of the Friedrich-Wilhelm award, and received the E-Plus award for his PhD thesis. He furthermore received the best paper award of the 2010 ITG Speech Communication Conference, the 2013 best student paper award at the IEEE Signal Processing Systems (SiPS) workshop and the 2014 IEEE Transactions on Communications Exemplary Reviewer Award. Since 2014, Dr. Schmalen also serves as guest lecturer at the University of Stuttgart. His research interests include channel coding, modulation formats and information theory for future optical communication systems.
\end{IEEEbiographynophoto}

\vfill


\begin{thebibliography}{10}

\bibitem{RengaswamyIZS16}
N.~Rengaswamy, L.~Schmalen, and V.~Aref, ``On the burst erasure correctability
  of spatially coupled {LDPC} ensembles,'' in {\em Proc. International Zurich
  Seminar on Communications}, pp.~155--159, 2016.

\bibitem{ARS_ISTC16}
V.~Aref, N.~Rengaswamy, and L.~Schmalen, ``Spatially coupled {LDPC} codes
  affected by a single random burst of errors,'' in {\em Proc. Intl. Symposium
  on Turbo Codes and Iterative Information Processing (ISTC)}, Sept. 2016.

\bibitem{GallagerIT62}
R.~G. Gallager, ``Low-density parity-check codes,'' {\em IRE Trans. Inform.
  Theory}, vol.~8, no.~1, pp.~21--28, 1962.

\bibitem{Richardson-MCT08}
T.~Richardson and R.~Urbanke, {\em Modern {C}oding {T}heory}.
\newblock Cambridge University Press, 2008.

\bibitem{Richardson2001design}
T.~Richardson, M.~Shokrollahi, and R.~Urbanke, ``Design of capacity-approaching
  irregular low-density parity-check codes,'' {\em IEEE Trans.\ Inform.\
  Theory}, vol.~47, no.~2, pp.~619--637, 2001.

\bibitem{Felstrom-it99}
A.~Felstr{\"o}m and K.~Zigangirov, ``Time-varying periodic convolutional codes
  with low-density parity-check matrix,'' {\em IEEE Trans.\ Inform.\ Theory},
  vol.~45, pp.~2181--2191, Sep 1999.

\bibitem{lentmaier2010thresholds}
M.~Lentmaier and G.~P. Fettweis, ``On the thresholds of generalized {LDPC}
  convolutional codes based on protographs,'' in {\em Proc.\ IEEE Int.\ Symp.\
  Inform.\ Theory (ISIT)}, pp.~709--713, IEEE, 2010.

\bibitem{Kudekar-it11}
S.~Kudekar, T.~Richardson, and R.~Urbanke, ``Threshold saturation via spatial
  coupling: Why convolutional {LDPC} ensembles perform so well over the
  {BEC},'' {\em IEEE Trans.\ Inform.\ Theory}, vol.~57, pp.~803--834, Feb~2011.

\bibitem{KudekarIT13}
S.~Kudekar, T.~Richardson, and R.~Urbanke, ``Spatially coupled ensembles
  universally achieve capacity under belief propagation,'' {\em IEEE Trans.\
  Inform.\ Theory}, vol.~59, no.~12, pp.~7761--7813, 2013.

\bibitem{Olmos-isit11}
P.~Olmos and R.~Urbanke, ``Scaling behavior of convolutional {LDPC} ensembles
  over the {BEC},'' in {\em Proc.\ IEEE Int.\ Symp.\ Inform.\ Theory (ISIT)},
  pp.~1816--1820, July 2011.

\bibitem{Stinner-isit14}
M.~Stinner and P.~Olmos, ``Analyzing finite-length protograph-based spatially
  coupled {LDPC} codes,'' in {\em Proc.\ IEEE Int.\ Symp.\ Inform.\ Theory
  (ISIT)}, pp.~891--895, June 2014.

\bibitem{Olmos-it15}
P.~Olmos and R.~Urbanke, ``A scaling law to predict the finite-length
  performance of spatially-coupled {LDPC} codes,'' {\em IEEE Trans.\ Inform.\
  Theory}, vol.~61, pp.~3164--3184, June 2015.

\bibitem{liva2012spatially}
G.~Liva, E.~Paolini, M.~Lentmaier, and M.~Chiani, ``Spatially-coupled random
  access on graphs,'' in {\em Proc.\ IEEE Int.\ Symp.\ Inform.\ Theory (ISIT)},
  pp.~478--482, IEEE, 2012.

\bibitem{Iyengar-icc10}
A.~Iyengar, M.~Papaleo, G.~Liva, P.~Siegel, J.~Wolf, and G.~Corazza,
  ``Protograph-based {LDPC} convolutional codes for correlated erasure
  channels,'' in {\em Proc.\ IEEE Int.\ Conf.\ Commun.}, pp.~1--6, May 2010.

\bibitem{Mori-corr15}
H.~Mori and T.~Wadayama, ``Band splitting permutations for spatially coupled
  {LDPC} codes enhancing burst erasure immunity,'' {\em CoRR},
  vol.~abs/1501.04394, 2015.
\newblock [Online]. Available: http://arxiv.org/abs/1501.04394.

\bibitem{ulHassan-isit14}
N.~Ul~Hassan, M.~Lentmaier, I.~Andriyanova, and G.~Fettweis, ``Improving code
  diversity on block-fading channels by spatial coupling,'' in {\em Proc.\ IEEE
  Int.\ Symp.\ Inform.\ Theory (ISIT)}, pp.~2311--2315, June~2014.

\bibitem{ulHassan-itw15}
N.~Ul~Hassan, I.~Andriyanova, M.~Lentmaier, and G.~Fettweis, ``Protograph
  design for spatially-coupled codes to attain an arbitrary diversity order,''
  in {\em Proc.\ IEEE Inform.\ Theory Workshop}, Oct 2015.

\bibitem{Jardel-comnet15}
F.~Jardel, J.~J. Boutros, M.~Sarkiss, and G.~{Rekaya-Ben Othman}, ``Spatial
  coupling for distributed storage and diversity applications,'' in {\em Proc.
  IEEE International Conference in Communications and Networking (ComNet)},
  (Hammamet, Tunisia), Nov. 2015.

\bibitem{Dedeoglu-ict15}
V.~Dedeoglu, F.~Jardel, and J.~J. Boutros, ``Spatial coupling of root-{LDPC}:
  Parity bits doping,'' in {\em Proc. IEEE International Conference on
  Telecommunications (ICT)}, (Sydney, Australia), Apr. 2015.

\bibitem{Boutros-tit10}
J.~J. Boutros, A.~{Guillen i Fabregas}, E.~Biglieri, and G.~Z{\'e}mor,
  ``Low-density parity-check codes for nonergodic block-fading channels,'' {\em
  IEEE Trans.\ Inform.\ Theory}, vol.~56, no.~9, pp.~4286--4300, 2010.

\bibitem{Jule-isit13}
A.~Jule and I.~Andriyanova, ``Performance bounds for spatially-coupled {LDPC}
  codes over the block erasure channel,'' in {\em Proc.\ IEEE Int.\ Symp.\
  Inform.\ Theory (ISIT)}, pp.~1879--1883, July 2013.

\bibitem{AndriyanovaBlackSea15}
I.~Andriyanova, N.~{Ul Hassan}, M.~Lentmaier, and G.~P. Fettweis, ``{SC-LDPC}
  codes over the block-fading channel: Robustness to a synchronisation
  offset,'' in {\em Proc. IEEE International Black Sea Conference on
  Communications and Networking (BlackSeaCom)}, pp.~97--101, May~2015.

\bibitem{tavares2007tail}
M.~B. Tavares, K.~S. Zigangirov, and G.~P. Fettweis, ``Tail-biting {LDPC}
  convolutional codes,'' in {\em 2007 IEEE International Symposium on
  Information Theory}, pp.~2341--2345, IEEE, 2007.

\bibitem{Orlitsky-isit02}
A.~Orlitsky, R.~Urbanke, K.~Viswanathan, and J.~Zhang, ``Stopping sets and the
  girth of {Tanner} graphs,'' in {\em Proc.\ IEEE Int.\ Symp.\ Inform.\ Theory
  (ISIT)}, 2002.

\end{thebibliography}
\end{document}